\documentclass[floatfix,superscriptaddress,aps,prl,twocolumn,reprint,longbibliography]{revtex4-1}
\usepackage[utf8]{inputenc}
\usepackage{amsmath}
\usepackage{amsfonts,dsfont}
\usepackage{amssymb,amsthm}
\usepackage{graphicx}
\usepackage[]{hyperref}
\hypersetup{
  colorlinks = true,
  urlcolor = magenta,
  linkcolor = red,
  citecolor = blue
}
\usepackage{mathpazo}
\usepackage{physics}
\usepackage[all]{hypcap}
\usepackage{url}
\usepackage{textcomp, complexity}
\usepackage{xcolor}

\usepackage[ruled, linesnumbered, procnumbered]{algorithm2e}
\let\oldnl\nl
\newcommand{\nonl}{\renewcommand{\nl}{\let\nl\oldnl}} 
\SetAlFnt{\normalsize}
\usepackage{booktabs}
\usepackage[capitalise,compress]{cleveref}
\crefname{section}{Sec.}{Section} 
\crefname{subsection}{subsection}{Subsection}

\newtheorem{thm}{Theorem}
\newtheorem{lemma}[thm]{Lemma}
\makeatletter

\def\@caption@fignum@sep{\ (Color online).\ }
\makeatother
\newcounter{parentnumber}

\crefname{thm}{Theorem}{Theorems}

\begin{document}
\title{Complexity phase diagram for interacting and long-range bosonic Hamiltonians}

\author{Nishad Maskara}\thanks{The two authors contributed equally.}
\affiliation{Department of Physics, California Institute of Technology, Pasadena, CA 91125, USA}
\affiliation{Joint Center for Quantum Information and Computer Science, NIST/University of Maryland, College Park, MD 20742, USA}
\author{Abhinav Deshpande}\thanks{The two authors contributed equally.}
\affiliation{Joint Center for Quantum Information and Computer Science, NIST/University of Maryland, College Park, MD 20742, USA}
\affiliation{Joint Quantum Institute, NIST/University of Maryland, College Park, MD 20742, USA}
\author{Adam Ehrenberg}
\affiliation{Joint Center for Quantum Information and Computer Science, NIST/University of Maryland, College Park, MD 20742, USA}
\affiliation{Joint Quantum Institute, NIST/University of Maryland, College Park, MD 20742, USA}
\author{Minh C.\ Tran}
\affiliation{Joint Center for Quantum Information and Computer Science, NIST/University of Maryland, College Park, MD 20742, USA}
\affiliation{Joint Quantum Institute, NIST/University of Maryland, College Park, MD 20742, USA}
\affiliation{Kavli Institute for Theoretical Physics, University of California, Santa Barbara, CA 93106, USA}
\author{Bill Fefferman}
\affiliation{Joint Center for Quantum Information and Computer Science, NIST/University of Maryland, College Park, MD 20742, USA}
\affiliation{Electrical Engineering and Computer Sciences, University of California, Berkeley, CA 94720, USA}
\author{Alexey V.\ Gorshkov}
\affiliation{Joint Center for Quantum Information and Computer Science, NIST/University of Maryland, College Park, MD 20742, USA}
\affiliation{Joint Quantum Institute, NIST/University of Maryland, College Park, MD 20742, USA}

\begin{abstract}
We classify phases of a bosonic lattice model based on the computational complexity of classically simulating the system.
We show that the system transitions from being classically simulable to classically hard to simulate as it evolves in time, extending previous results to include on-site number-conserving interactions and long-range hopping.
Specifically, we construct a ``complexity phase diagram'' with ``easy'' and ``hard'' phases, and derive analytic bounds on the location of the phase boundary with respect to the evolution time and the degree of locality.
We find that the location of the phase transition is intimately related to upper bounds on the spread of quantum correlations and protocols to transfer quantum information.
Remarkably, although the location of the transition point is unchanged by on-site interactions, the nature of the transition point changes dramatically.
Specifically, we find that there are two kinds of transitions, sharp and coarse, broadly corresponding to interacting and noninteracting bosons, respectively.
Our work motivates future studies of complexity in many-body systems and its interplay with the associated physical phenomena.
\end{abstract}

\maketitle

A major effort in quantum computing is to find examples of quantum speedups over classical algorithms, despite the absence of general principles characterizing such a speedup.
The study of classical simulability of quantum systems evolving in time allows one to identify features underlying a quantum advantage.
Studying the classical simulability of both quantum circuits \cite{Valiant2002,Terhal2002,Terhal2002a,Aaronson2004,Jozsa2008,Ni2012,Lloyd1995,Deutsch1995,Bremner2002,Aaronson2011,Bremner2011,Fefferman2016,Bremner2016,Bremner2017,Bermejo-Vega2018,Hangleiter2018,Haferkamp2019} and Hamiltonians \cite{Childs2011,Bouland2016}, especially under restrictions such as spatial locality \cite{Deshpande2018,Muraleedharan2018}, allows one to understand the classical-quantum divide in terms of their respective computational complexity.

\begin{figure}[h!]
\includegraphics[width=0.8\linewidth]{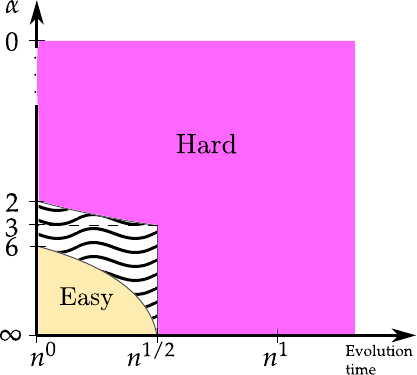}
\caption{A slice of the complexity phase diagram for the long-range bosonic Hamiltonian in 2D with $n$ bosons when the initial inter-boson spacing is $L$\,$=$\,$\Theta(\sqrt{n})$ (see \cite{note_littleo}).
Colors represent whether the sampling problem is easy (yellow), hard (magenta), or not currently known (hatched).
The $X$-axis parametrizes the evolution time as a polynomial function of $n$, and the $Y$-axis is $\alpha$, the exponent characterizing the long-range nature of the hopping Hamiltonian (with scale $y$\,$=$\,$1/\sqrt{\alpha}$ except for the point $\alpha$\,$=$\,$0$).
} \label{fig_phasediag}
\end{figure}

In this work, we characterize the worst-case computational complexity of simulating time evolution under bosonic Hamiltonians and study a dynamical phase transition in approximate sampling complexity \cite{Deshpande2018,Muraleedharan2018}. 
Previous work \cite{Deshpande2018} studied free bosons with nearest-neighbor hopping but did not consider the robustness of the transition to perturbations in the Hamiltonian, a crucial question in the study of any phase transition.
We generalize Ref.\,\cite{Deshpande2018} to include number-conserving interactions and long-range hops and conclude that the phase transition is indeed robust.
These kinds of interactions are ubiquitous in experimental implementations of hopping Hamiltonians with ultracold atoms and superconducting circuits \cite{Norcia2018,Neill2018}.
Long-range hops which fall off as a power law are also native to several architectures \cite{Saffman2010,Britton2012,Yao2012,Yan2013,Douglas2015}.
We study the location of the phase transition and its dependence on various system parameters, constructing a \textit{complexity phase diagram}, a slice of which is presented in \cref{fig_phasediag}.

One new insight from this work is the discovery of two different kinds of complexity phase transitions, sharp and coarse, in the context of dynamical quantum systems.
Sharp and coarse transitions are common in probabilistic graph theory \cite{Janson2000}, and are reminiscent of I- and II-order phase transitions in many-body physics.
Specifically, in interacting systems, which are universal for quantum computation, we find coarse transitions in 1D and sharp transitions in higher dimensions.
Further, our results suggest that for noninteracting systems, which are not believed to be universal for quantum computation, the transition is coarse in all dimensions.

\emph{Setup and summary of results.---}
Consider a system of $n$ bosons hopping on a cubic lattice of $m$ sites in $D$ dimensions with real-space bosonic operators $a_j$.
We let $m$\,$=$\,$\Theta(n^\beta)$ (see \cite{note_littleo}) and assume sparse filling: $\beta$\,$\geq$\,$1$.
The Hamiltonian $H$\,$=$\,$\sum_{i,j}J_{ij}(t) a_i^\dag a_j+\mathrm{h.c.}+ \sum_i f(n_i)$ has time-dependent hopping terms bounded by a power-law $\abs{J_{ij}(t)}$\,$\leq$\,$1/d(i,j)^\alpha$ and on-site interactions $f(n_i)$.
The parameter $\alpha$ governs the degree of locality.
When $\alpha$\,$=$\,$0$, the system has all-to-all couplings, while $\alpha$\,$\rightarrow$\,$\infty$ corresponds to nearest-neighbor hops.
The on-site terms $J_{ii}(t)$ can be large, and the interaction strength is $\abs{f(n_i)} \sim V$.
For concreteness, our hardness results are derived using a Bose-Hubbard interaction $f(n_i)$\,$=$\,$Vn_i(n_i-1)/2$, but the timescales we present are valid for generic on-site interactions \cite{Childs2013}.
The bosons in the initial states considered are sparse and well-separated.
Specifically, partition the lattice into $K$ clusters $C_1,\ldots,C_K$ containing $b_1,\ldots,b_K$ initial bosons respectively, such that $b$\,$:=$\,$\max{b_i}$\,$=$\,$O(1)$ does not scale with lattice size.
Define the width $L_i$ of a cluster $C_i$ as the minimum distance between a site outside the cluster and an initially occupied site inside the cluster 
and let $L$\,$=$\,$\min_i{L_i}$.
While this can be done for any initial state, choosing a good clustering so the separations $L_i$ are large may be difficult.
As in Ref.\,\cite{Deshpande2018}, we consider states with $L$\,$=$\,$\Theta(n^{(\beta-1)/D})$.

The computational task of approximate sampling is to simulate projective measurements of the time-evolved state in the local boson-number basis.
The \textit{approximate sampling complexity} measures the classical resources needed to produce samples from a distribution $\tilde{\mathcal{D}}$ that is $\epsilon$\,$=$\,$ O(1/\poly(n))$-close in total variation distance to the target distribution $\mathcal{D}$ \cite{note_weaker}.
Sampling from a distribution $\tilde{\mathcal{D}}$ satisfying the above takes runtime $T(n,t)$ in the worst case on a classical computer, where $t$ is the evolution time.
Like thermodynamic quantities, the complexity is defined asymptotically as $n$\,$\to$\,$\infty$, so we consider the scaling of $T$ along a curve $t(n)$.
For any curve $t(n)$, sampling is \emph{easy} if there exists a polynomial-runtime classical algorithm, meaning $T(n,t(n))$\,$=$\,$O(n^k)$ for constant $k$, or \emph{hard} if such an algorithm cannot exist.
Since the problem is either easy or hard for a particular function $t(n)$, there is always a transition in complexity as opposed to a smooth crossover.
The \emph{transition timescale} $t_*(n)$ is a function such that for any timescale $t = o(t_*)$ the problem is easy and for any timescale $t$\,$=$\,$\omega(t_*)$ it is hard.
For reasons that will become clear, we consider the scaling $t(n)$\,$=$\,$cn^\gamma$ and place bounds on the location of the transition: $t_\mathrm{easy}(n)$\,$\equiv$\,$c_\mathrm{easy}n^{\gamma_\mathrm{easy}}$\,$\leq$\,$t_*(n)$\,$\leq$\,$t_\mathrm{hard}(n)$\,$\equiv$\,$c_\mathrm{hard}n^{\gamma_\mathrm{hard}}$, where $t_{*}(n)$\,$\equiv$\,$c_{*}n^{\gamma_{*}}$.

We find that the transition comes in two types, which we call ``sharp'' and ``coarse'' (\cref{fig_sharpcoarse}).
For sharp transitions, the optimal exponents overlap ($\gamma_\mathrm{hard}^{\mathrm{opt}}$\,$=$\,$\gamma_\mathrm{easy}^{\mathrm{opt}}$) and the transition occurs in the coefficients ($c_\mathrm{hard}^{\mathrm{opt}}$\,$>$\,$ c_\mathrm{easy}^{\mathrm{opt}}$).
For coarse transitions, we instead have $\gamma_\mathrm{hard}^{\mathrm{opt}}$\,$>$\,$ \gamma_\mathrm{easy}^{\mathrm{opt}}$ (see \cite{note_precisedef,Janson2000} for more precise definitions).
An example of a sharp transition is when the transition timescale is $t_*$\,$=$\,$2n$, so that the problem is easy for all times $t$\,$\leq$\,$1.99n$ and hard for all times $t$\,$\geq$\,$2.01n$.
An example of a coarse transition is when the transition timescale is $t_*$\,$=$\,$\Theta(n\log n)$, so that the problem is easy for all times $t$\,$\leq$\,$cn$ and hard for all times $t$\,$\geq$\,$cn^{1.01}$.

\begin{figure}
\includegraphics[width=\linewidth]{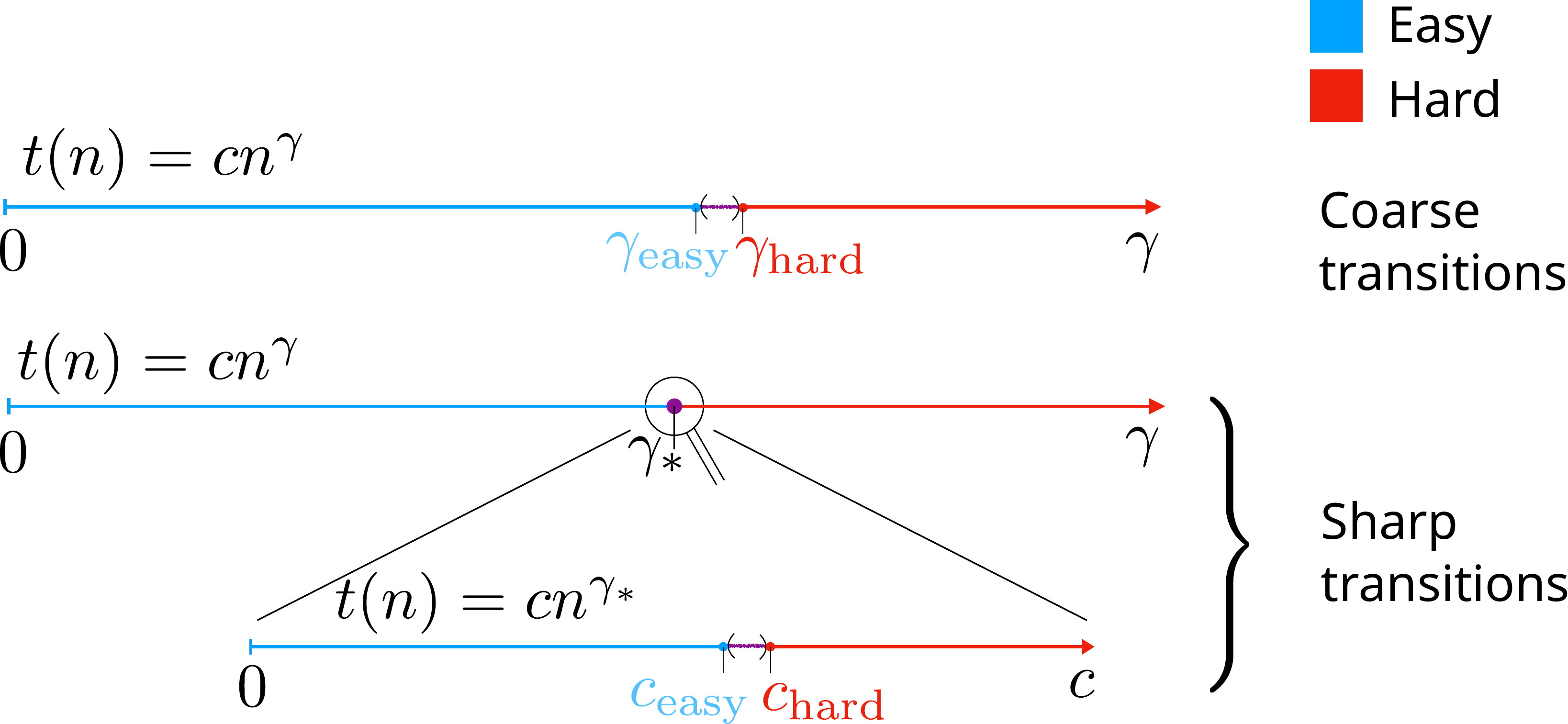}
\caption{Schematic for (a) coarse and (b) sharp transitions under the scaling $t(n)$\,$=$\,$cn^\gamma$.
For coarse transitions, the transition timescale lies between two exponents $\gamma_\mathrm{easy}$ and $\gamma_\mathrm{hard}$.
For sharp transitions, the transition timescale is at an exponent $\gamma_*$, and the complexity of the problem depends on whether the prefactor $c$ is smaller or larger than $c_*$.} \label{fig_sharpcoarse}
\end{figure}

We summarize our main results in \cref{thm_easy,thm_hard}.
The easiness result comes from applying classical algorithms for quantum simulation, and depend on Lieb-Robinson bounds on information transport \cite{Lieb1972,Hastings2006,Gong2014,Foss-Feig2015,Tran2019}.
The hardness results come from reductions to families of quantum circuits for which efficient approximate samplers cannot exist, modulo widely believed conjectures in complexity theory \cite{Aaronson2011,Bermejo-Vega2018,Hangleiter2018,Haferkamp2019}, and from fast protocols to transmit quantum information across long distances \cite{Guo2019a,Tran2020}.

\begin{thm}[Easiness result]\label{thm_easy}
For $\alpha $\,$>$\,$ D+1$, and for all $V$, including $V $\,$=$\,$ o(1)$ and $V=\omega(1)$, we have $t_{\mathrm{easy}} $\,$=$\,$ \Omega(n^{\gamma_\mathrm{easy}})$, with
\begin{align}\label{eq_easytimescale}
\gamma_{\mathrm{easy}} = \frac{\beta-1}{D} \times \frac{\alpha - 2D}{\alpha-D} - \frac{1}{\alpha-D},
\end{align}
and $t_{\mathrm{easy}} $\,$=$\,$ \Omega(\log n)$ if $\gamma_\mathrm{easy} $\,$<$\,$ 0$.
\end{thm}
\cref{thm_easy} is valid for any form of the on-site interaction $f(n_i)$ and features the same timescale irrespective of the interaction strength.
In the nearest neighbor limit $\alpha$\,$\rightarrow$\,$\infty$, this reproduces the timescale $t_\mathrm{easy} $\,$=$\,$\Omega(n^{(\beta-1)/D})$\,$=$\,$\Omega(L)$, which corresponds to the timescale when interference between clusters become relevant \cite{Deshpande2018}.
When $\beta$\,$=$\,$1$, $\gamma_\mathrm{easy}$ becomes negative, and we instead have $t_\mathrm{easy} $\,$=$\,$ \Omega(\log n)$, matching the result in Ref.\,\cite{Muraleedharan2018} for nearest-neighbor hops and closely distributed initial states with $L$\,$=$\,$O(1)$.

\begin{thm}[Hardness result] \label{thm_hard} When $\alpha$\,$\geq$\,$D/2$, $V$\,$=$\,$\Omega(1)$, and $D$\,$\geq$\,$2$, the hardness timescale is $t_\mathrm{hard}$\,$=$\,$O(n^{\gamma^\mathrm{I}_\mathrm{hard}})$, where
\begin{align}
\gamma^\mathrm{I}_\mathrm{hard} = 
\begin{cases}
\frac{\beta-1}{D} \min[1,\alpha - D], & \alpha > D \\
0, & \alpha \in [\frac{D}{2},D].
\end{cases}
\label{eq_type1}
\end{align}
When $\alpha$\,$<$\,$D$\,$/$\,$2$, $V$\,$=$\,$o(1)$, or $D$\,$=$\,$1$, the timescale is $t_\mathrm{hard}$\,$=$\,$O(n^{\gamma^\mathrm{II}_\mathrm{hard}})$, where
\begin{align}
\gamma^\mathrm{II}_\mathrm{hard} = \delta +
\begin{cases}
\frac{\beta-1}{D} \min\left[1 + \frac{O(\log (V + 1))}{\log n} , \alpha - D \right], \quad \alpha > D \\
0, \quad \alpha \in [\frac{D}{2},D] \\
\frac{\beta}{D}\left(\alpha - \frac{D}{2}\right), \quad \alpha < D/2
\end{cases}
\label{eq_type2}
\end{align}
for an arbitrarily small $\delta>0$. 
\end{thm}
We examine the various limits: $\alpha$\,$\rightarrow$\,$\infty$ (nearest-neighbor), $\alpha $\,$\to$\,$0$ (all-to-all connectivity), $V$\,$\to$\,$0$ (free bosons), and $V$\,$\to$\,$\infty$ (hardcore bosons).
First, when $\alpha$\,$\rightarrow$\,$\infty$, the hardness timescale upper bound is $O(L)$ in all cases except when $V$\,$\rightarrow$\,$\infty$, $D$\,$=$\,$1$, which we discuss below.
The timescale $O(L)$ again corresponds to the distance $L$ between clusters, matching the corresponding bound in \cref{eq_easytimescale}, and therefore pinning the transition timescale to $\Theta(L)$.
In the opposite limit when the model is sufficiently long-range ($\alpha $\,$<$\,$ D/2$), the role of the dimension is unimportant, giving $\gamma_\mathrm{hard}$\,$<$\,$0$ in all cases.
This suggests a hardness timescale close to 0, signifying the immediate onset of hardness.
Next, free bosons ($V$\,$=$\,$o(1)$) have almost the same hardness timescale (up to arbitrarily small $\delta$\,$>$\,$0$) as interacting bosons everywhere in the phase diagram.
This shows that the location of the complexity phase transition is robust to the presence of interactions.
In fact, the interaction strength $V$ does not affect the timescale except in the 1D nearest-neighbor hardcore limit.
In this case, there is no hardness regime, as seen through the divergence of $\gamma^\mathrm{II}_\mathrm{hard}$ in \cref{eq_type2} when $\alpha,V$\,$\to$\,$ \infty$.
This is because the model maps to that of free fermions, or equivalently, matchgate circuits, which are easy to simulate at all times \cite{Valiant2002,Terhal2002}.
We now outline the proofs of our results, whose details may be found in Ref.\,\cite{SM}.

\emph{Easy-sampling timescale.---} To derive $t_\mathrm{easy}$, we give an efficient sampling algorithm.
The algorithm performs time evolution on each cluster $C_i$ separately.
This takes polynomial time in the number of basis states, which is $\binom{\abs{C_i} + b_i - 1}{b_i} $\,$=$\,$ O(\abs{C_i}^{b_i})$ and hence polynomial in $n$ when $b_i$\,$=$\,$O(1)$.
This product-state approximation of the exact time-evolved state $\ket{\psi(t)} $\,$=$\,$ U_t \ket{\psi(0)}$ is achieved by decomposing the propagator $U_t$ via a spatial decomposition scheme for quantum simulation \cite{Haah2018,Tran2019} that we call the HHKL decomposition.
We complete the derivation of the easiness timescale by showing that the approximation is good for times $t$\,$<$\,$O(t_\mathrm{easy})$.

Here, we briefly present the HHKL decomposition, which is powerful but remarkably simple.
Let $H_R$ be the sum over all terms in the Hamiltonian supported completely in region $R$ and implicitly let $XY $\,$=$\,$ X$\,$\cup$\,$Y$ represent the union of regions.
The forward-time propagator is $U_{t_0,t_1}^{R} $\,$=$\,$ \mathcal{T} \exp(- i \int_{t_0}^{t_1} H_R(s) ds)$.
The decomposition scheme approximates a unitary acting on region $XYZ$ (where $Y$ separates regions $X$ and $Z$) by forward evolution on $YZ$, backward evolution on $Y$, and forward evolution on $XY$: $U_{XYZ}$\,$\approx$\,$U_{XY}(U_Y)^\dag U_{YZ}$.
The operator norm error made by this approximation is \cite{Tran2019} $O \left((e^{vt} - 1) \Phi(X)({\ell^{-\alpha + D + 1}}\right.$\,$+$\,$\left.e^{-\ell})\right)$, where $v$\,$>$\,$0$ is a characteristic velocity, $\Phi(X)$ is the area of the boundary of $X$, and $\ell$ is the minimum distance between any pair of sites in $X$ and $Z$.
The error is small for times $t$ shorter than the time it takes for information to propagate from $X$ to $Z$.

The velocity $v$ of information propagation is also known as a Lieb-Robinson velocity and is determined by the operator norm of terms in the Hamiltonian which couple different sites \cite{Hastings2006}.
Since bosonic operators have unbounded operator norm, this could result in an unbounded velocity \cite{Eisert2009}.
However, because of boson number conservation under the Hamiltonian, the dynamics is fully contained in the $n$-boson subspace, within which the operator norm of each term is $O(n)$.
While free bosons ($V$\,$=$\,$0$) behave as in the single-particle subspace, implying the Lieb-Robinson velocity is $O(1)$, in the interacting case, an $O(n)$ Lieb-Robinson velocity would cause the asymptotic easiness timescale to vanish ($t_{\mathrm{easy}}$\,$\rightarrow$\,$0$).

Nevertheless, the easiness timescale we derive is independent of $V$ for a clustered initial state.
Intuitively, at short times each boson is well-localized within its original cluster.
Therefore, the relevant subspace has at most $b$ bosons in each cluster $C_i$.
Truncating the Hilbert space to allow only $b+1$ bosons per cluster is therefore a good approximation at short times \cite{SM,Peropadre2017}, and the truncation error vanishes in the asymptotic limit.
The modified Hamiltonian $H'$ after truncation has terms with norm only $O(b)$, giving an effective Lieb-Robinson velocity $v$\,$=$\,$O(b)$\,$=$\,$O(1)$ for states close to the initial state \cite{Note2}.
For this modified Hamiltonian, we apply the HHKL decomposition to bound the error caused by simulating each cluster separately.
Once the error has been calculated, the timescale immediately follows by solving $\epsilon(t) $\,$=$\,$O(1)$ for $t$\,$=$\,$t_{\mathrm{easy}}$, which is a lower bound on the transition timescale $t_*$.
In Ref.\,\cite{SM}, we give the full dependence of $t_{\mathrm{easy}}$ on various system parameters, including the filling fraction of bosons.

\emph{Sampling hardness timescale.---}
To derive $t_\mathrm{hard}$, we give protocols to simulate quantum circuits by setting the time dependent parameters $J_{ij}(t)$ of the long-range bosonic Hamiltonian.
This implies sampling is worst-case hard after time $t_\mathrm{hard}$.
Specifically, if a general sampling algorithm exists for times $t$\,$\geq$\,$ t_\mathrm{hard}$, we prove this algorithm can also simulate hard instances of boson sampling \cite{Aaronson2011} when interactions are weak, and quantum circuits that are hard to simulate \cite{Bermejo-Vega2018} when interactions are strong.

In the interacting case, our reduction from universal quantum computation to a long-range Hamiltonian hinges on implementing a universal gate set.
Using a dual-rail encoding to encode a qubit in two modes of each cluster $C_i$, we show in Ref.\,\cite{SM} how to implement arbitrary single-qubit operations in $O(1)$ time and controlled-phase gates \cite{Underwood2012} between adjacent clusters in a time that depends on their spacing $L$.
The two-qubit gate uses free particle state-transfer as a subroutine \cite{Guo2019a,Tran2020} to bring adjacent logical qubits near each other.
We implement the constant-depth circuit of Ref.\,\cite{Bermejo-Vega2018}, which consists only of nearest-neighbor gates between qubits in a 2D grid.
The total time for hardness under this scheme takes time $O(\min$\,$[L$\,$,$\,$L^{\alpha - D}])$ when $\alpha $\,$>$\,$ D$ and $O(1)$ when $\alpha$\,$\in$\,$[D/2,D]$.
In 1D, simulating a 2D circuit introduces extra overhead.
Nevertheless, we can recover the same timescale up to an infinitesimal $\delta$\,$>$\,$0$ in the exponent by only encoding $n^\delta$ logical qubits.
For hardcore bosons, the above scheme mentioned does not work and the entangling gate is constructed differently, and features an easiness result for the 1D nearest-neighbor case.
Lastly, when $\alpha$\,$<$\,$D$\,$/$\,$2$, state transfer takes time $o(1)$ but the time for an entangling gate is $O(1)$.
We can still achieve coarse hardness for time $o(1)$ by mapping the system onto free bosons, which we now come to.

In the noninteracting case, we implement the boson sampling scheme of Ref.\,\cite{Aaronson2011}, which showed that a Haar-random unitary applied to $m$ sites containing $n$ bosons gives a hard-to-sample state.
It also gave an $O(n\log m)$-depth decomposition of a linear-optical unitary in the circuit model without spatial locality.
We give a faster implementation for the continuous-time Hamiltonian model, which can include simultaneous noncommuting terms but imposes spatial locality, a result of independent interest \cite{SM}.
Specifically, we show that most linear-optical states of $n$ bosons on $m$ sites can be constructed in time $\min$\,$[O(nm^{1/D})$\,$,$\,$\tilde{O}(nm^{\alpha/D - 1/2}])$, which is faster than the circuit model when $\alpha$\,$<$\,$D$\,$/$\,$2$.
This result also uses free-particle state transfer as a subroutine.
As in the 1D interacting case, we can implement the reduction on a polynomially growing number of bosons $n^\delta$, resulting in the timescale of \cref{eq_type2} for free bosons.
This result resolves an important conceptual question posed by Ref.\,\cite{Deshpande2018} for the noninteracting, nearest-neighbor case by closing the gap between $t_\mathrm{easy}$ and $t_\mathrm{hard}$.
In this limit, the transition timescale is at $\Theta(L/v)$, both with and without interactions, showing that the algorithm of Ref.\,\cite{Deshpande2018} is optimal and that the presence of interactions does not change the phase diagram.

\emph{Sharp and coarse transitions---}
We now discuss the role of the term $\delta$ in \cref{eq_type2}.
This infinitesimal is suggestive of a coarse transition, because it ensures $\gamma_{\mathrm{hard}}$\,$>$\,$\gamma_*$ \cite{note_coarsecaveat}.
Therefore, our results suggest that the main difference between the interacting and noninteracting models is the type of transition induced.
In the presence of interactions and in dimensions 2 and above, the bounds on the timescale in the nearest-neighbor limit coincide at $t_*$\,$=$\,$\Theta(L)$, proving sharpness of the transition.
In the 1D/noninteracting case, however, our results suggest that the transition is coarse.
In the noninteracting case specifically, our work implies that either the transition is coarse, or there exists a constant-depth boson sampling circuit for which approximate sampling is classically hard.
Both of these possibilities are interesting in their own right, but we believe the first is more likely to be true.

For $D$\,$=$\,$1$, the transition is coarse when $\alpha$\,$\to$\,$\infty$.
One way of understanding this is from tensor-network algorithms like matrix product states to simulate the problem, which work well for systems with area-law entanglement.
For this specific case ($\alpha$\,$\to$\,$\infty$, $D$\,$=$\,$1$), we can use the fact that time evolution is classically simulable for any logarithmic time \cite{Osborne2006} by exploiting the connection to matrix product states.
In our setup, this translates to an easiness timescale of $cL$\,$\mathrm{polylog}L$ for any $c$ and any polylogarithmic function, which is consistent with the hardness timescale being $L^{1+\delta}$ for any $\delta$\,$>$\,$0$.
Therefore the transition is coarse, and $t_{\mathrm{hard}}$\,$=$\,$\omega(n^{\gamma_{\mathrm{easy}}})$.
However, if $D$\,$\geq$\,$2$, this argument breaks down because tensor-network contraction takes time exponential in the system size in the worst case \cite{Schuch2007a}, and there are known examples of constant-depth 2D circuits that are hard to simulate \cite{Bermejo-Vega2018,Bremner2017}.

\textit{Outlook.---}
We have mapped out the complexity of the long-range Bose-Hubbard model as a function of the particle density $\beta$, the degree of locality $\alpha$, the dimensionality $D$, and the evolution time $t$.
A particularly interesting open question concerns the regions of the phase diagram without definitive easiness/hardness results.
These gaps are closely related to open problems in other areas of many-body physics and quantum computing.
In the nearest-neighbor limit, there is no gap between $t_\mathrm{easy}$ and $t_\mathrm{hard}$.
When $\alpha$ is finite, closing the gap is closely tied to finding state-transfer protocols which saturate Lieb-Robinson bounds.
Stronger Lieb-Robinson bounds can increase $t_\mathrm{easy}$, and faster state-transfer will reduce $t_\mathrm{hard}$, as evidenced by the improvement over the previous version of this manuscript due to results from Ref.\,\cite{Tran2020}.
These observations show that studying complexity phase transitions provides a nice testbed for, and gives an alternative perspective on results pertaining to the locality of quantum systems.

It is illuminating to study the approach to the transition from either regime.
On the easiness side, the error made in the HHKL decomposition algorithm grows with time until it reaches $O(1)$ at time $t_*$.
On the hardness side, the transition behaves qualitatively differently for sharp and coarse transitions.
For coarse transitions, as the evolution time is reduced to $t_*$, the number of encoded logical qubits shrinks as $n^\delta$, where $\delta $\,$\rightarrow$\,$ 0$ as $t$\,$\rightarrow$\,$ t_*$.
This illustrates that while the problem is still asymptotically hard as $n $\,$\rightarrow$\,$ \infty$, one needs to go to higher boson numbers $n$ to achieve the same computational complexity.
For sharp transitions on the other hand, the number of encoded logical qubits seems to behave as $\delta$\,$\poly(n)$.
This illustrates a physical difference between the two types of computational phase transitions near the transition point, hinting at a rich variety of possibly undiscovered complexity phase diagrams.

Our results can be easily adapted to a wide range of experimentally and theoretically interesting Hamiltonians.
Spin Hamiltonians naturally map onto our model in the hardcore limit.
Fermionic systems with nearest-neighbor interactions can also be incorporated by performing the mapping described in Ref.\,\cite{Verstraete2005}.
Our model is also relevant to cold atom experiments that have been proposed as candidates for observing quantum computational supremacy \cite{Muraleedharan2018,Norcia2018,Bermejo-Vega2018,Neill2018}, especially in the nearest-neighbor limit.
The power-law hopping $1/r^\alpha$ can be engineered to directly implement the classes of Hamiltonians we study.
We can also virtually couple our band of interest to another with a quadratic band edge to implement exponentially decaying hopping \cite{Douglas2015,Chu2019,Gadway2015}.
Doing this simultaneously with multiple detunings approximates a power-law with high accuracy as a sum of exponentials \cite{Crosswhite2008}.
In the hardcore limit, the long-range hops translate to long-range interactions between spins, which model quantum-computing platforms such as Rydberg atoms and trapped ions \cite{Saffman2010,Bernien2017,Barredo2016,Korenblit2012,Islam2013}.
Therefore, the Hamiltonian we study models various physically interesting situations, both in the several limiting cases ($\alpha$\,$\rightarrow$\,$\infty$, $V$\,$\rightarrow$\,$0$, $V$\,$\rightarrow$\,$\infty$) as well as in the general case of finite nonzero $\alpha$ and $V$.
Furthermore, our methods also work for general number-conserving Hamiltonians, for example, long-range density-density interactions $K_{ij}(t)$\,$n_in_j$ with nearest-neighbor hops.
The only effect on the easiness times is to modify the Lieb-Robinson velocity to $v$\,$=$\,$O(b^2)$.

Our model can also describe a distributed modular quantum network when $V$ can vary spatially.
Specifically, a module of qubits can be represented by hardcore bosons ($V$\,$\rightarrow$\,$\infty$), while photonic communication channels linking distant modules can be represented by sites with $V$\,$=$\,$0$ separating the modules.
As in quantum networks, our hardness times in the nearest-neighbor regime are dominated by gates between nodes, while operations within a single node are free. 

There is also an intriguing connection between the (dynamical) phase transitions we study as a function of time and (equilibrium) phase transitions as a function of temperature.
Interacting bosons in 2D and above feature sharp transitions, falling into one ``universality class'' separate from that of free bosons and 1D.
This is reminiscent of equilibrium phase transitions where the universality class depends strongly on the dimension and on the nature of interactions.
This connection may be further investigated by studying complexity phase transitions in thermal states as a function of temperature \cite{Bravyi2008a,Poulin2009,Chiang2010,Temme2011,Kastoryano2016,Brandao2019a,Harrow2019a,Kuwahara2019b,Kato2019}.

\begin{acknowledgments}
\emph{Acknowledgments.---}
We thank Michael Foss-Feig, James Garrison, Dominik Hangleiter, Rex Lundgren, and Emmanuel Abbe for helpful discussions, the anonymous Referee for their valuable comments, and to the authors of Ref.\,\cite{Guo2019a} for sharing their results with us.
N.~M., A.~D., M.~C.~T., A.~E., and A.~V.~G.~acknowledge funding by DoE ASCR FAR-QC (award No. DE-SC0020312), NSF PFCQC program, DoE BES Materials and Chemical Sciences Research for Quantum Information Science program (award No. DE-SC0019449), the DoE ASCR Quantum Testbed Pathfinder program (award No. DE-SC0019040), AFOSR MURI, AFOSR, ARO MURI, ARL CDQI, and NSF PFC at JQI.
M.~C.~T.\ also acknowledges support under the NSF Grant No.\ PHY-1748958 and from the Heising-Simons Foundation.
N.~M.\ also acknowledges funding from the Caltech SURF program.
A.~E.\ also acknowledges funding from the DoD.
B.~F.\ is funded in part by AFOSR YIP No.\ FA9550-18-1-0148 as well as ARO Grants No.\ W911NF-12-1-0541 and No.\ W911NF-17-1-0025, and NSF Grant No.\ CCF-1410022.
\end{acknowledgments}

\begin{widetext}
\renewcommand{\thesection}{S\arabic{section}}
\renewcommand{\theequation}{S\arabic{equation}}
\renewcommand{\thefigure}{S\arabic{figure}}
\setcounter{equation}{0}
\setcounter{figure}{0}

\section*{Supplemental material}
\textbf{Abstract:} In this Supplemental Material, we give the full proofs of \cref{thm_easy,thm_hard}.

\section{Section I: \qquad Approximation error under HHKL decomposition}
We first argue why it is possible to apply the HHKL decomposition lemma to $H'$ with a Lieb-Robinson velocity of order $O(1)$.
As mentioned in the main text, $H'$ is a Hamiltonian that lives in the truncated Hilbert space of at most $b +1$ bosons per cluster.
Let $Q$ be a projector onto this subspace.
Then $H' = QHQ$.
Time-evolution under this modified Hamiltonian $H'$ keeps a state within the subspace since $\comm{e^{-iQHQt}}{Q} = 0$.

The Lieb-Robinson velocity only depends on the norm of terms in the Hamiltonian which couple lattice sites.
On-site terms do not contribute, which can be seen by moving to an interaction picture \cite{Foss-Feig2015,Hastings2006}. 
Therefore, since no state has more than $b+1$ bosons on any site within the image of $Q$, the maximum norm of coupling terms in $H'$ is $\norm{Qa_i^\dag a_jQ} \leq b + 1$. Therefore, the Lieb-Robinson velocity is at most $O(b)$ instead of $O(n)$, and we can apply the HHKL decomposition to the evolution generated by the truncated Hamiltonian $H'$.
We now prove that the error made by decomposing the evolution due to $H'$ is small.
\begin{lemma}[Decomposition error for $H'$] \label{lem_easyQHQ}
For all $V$ and $\alpha $\,$>$\,$ D+1$, the error incurred (in 2-norm) by decomposing the evolution due to $H'$ for time $t$ is
\begin{align}
 \epsilon(t) & \leq O \left(K(e^{v{t_1}} - 1) (\ell^{-\alpha+D+1} + e^{- \ell})\sum_{j=0}^{N-1} (r_0+j\ell)^{D-1}\right),
\end{align}
where $N $\,$=$\,$ t/t_1$ and $\ell $\,$\leq$\,$ L/N$ can be chosen to minimize the error, and $r_0$ is the radius of the smallest sphere containing the initially occupied bosons in a cluster.
\end{lemma}

\begin{figure}[b]
\includegraphics[width=0.5\linewidth]{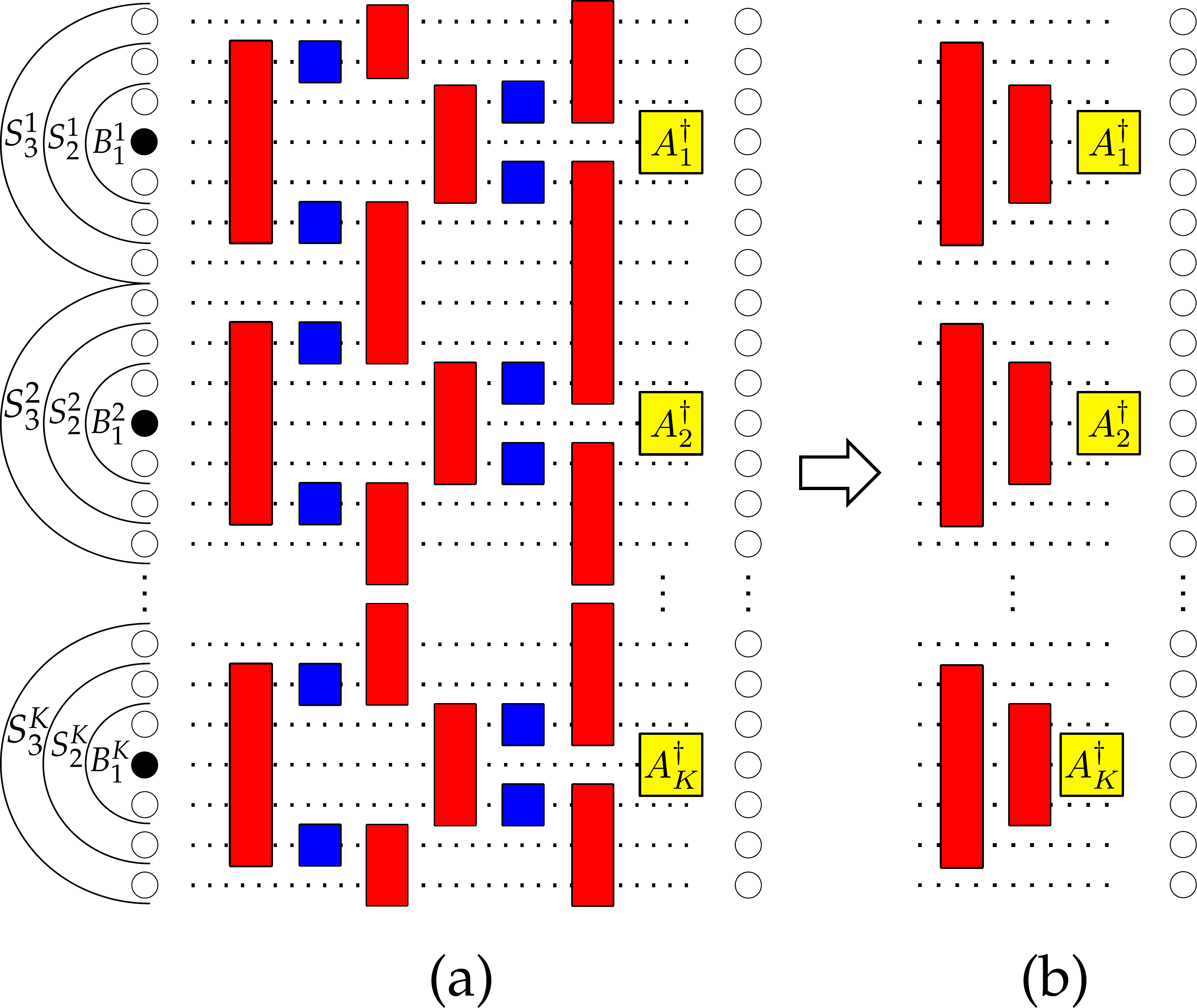}
\caption{(a) Decomposition of the first two steps of the unitary evolution followed by (b) pushing the commuting terms past $A_i^\dagger$ (the product of all initial creation operators in a cluster $i$) to the vacuum.
Red boxes represent forward evolution and blue boxes backward evolution in time.} \label{fig_decompos}
\end{figure}

The sketch of the proof is as follows: recall that within each cluster $C_i$, there is a group of bosons initially separated from the edge of the cluster by a region of width $L_i$.
Naive application of the HHKL decomposition for the long-range case results in a timescale $t_{\mathrm{easy}}$\,$\sim$\,$\log(L)$, because of the exponential factor $(e^{vt}-1)$ in the error.
To counter this, we apply the HHKL decomposition in small time-steps $t_1$.
Thus, within each time-step, the exponential factor can be approximated as $e^{v t_1}-1 $\,$\approx$\,$v t_1$, turning this exponential dependence into a polynomial one at the cost of an increased number of time-steps.

The first two time-steps are depicted pictorially in \cref{fig_decompos}, and illustrate the main ideas. 
The full propagator acting on the entire lattice is decomposed by applying the HHKL decomposition $K$ times, such that two of every three forward and reverse time-evolution operators commute with all previous operators by virtue of being spatially disjoint, allowing them to be pushed through and act identically on the vacuum.
The remaining forward evolution operator effectively spreads out the bosonic operators by distance $\ell$.
The error per time-step is polynomially suppressed by $O(\ell^{-\alpha+D+1} + e^{-\gamma \ell})$.

While it reduces the exponential factor to a polynomial one, using time-slices comes at the cost of extra polynomial factors, originating from the sum over boundary terms $\sum_{j=0}^{N-1} (r_0 + j\ell)^{D-1}$. 

\begin{proof}[Proof of \cref{lem_easyQHQ}]
Let the initial positions of the bosons be denoted by $(\mathsf{in}_1$\,$,$\,$\ldots$\,$,$\,$\mathsf{in}_n)$.
The initial state is $\ket{\psi(0)} = a_{\mathsf{in}_1}^\dag \ldots a_{\mathsf{in}_n}^\dag \ket{0}$
As before, the first two time-steps are illustrated in Fig.\,\ref{fig_decompos}.
Within each cluster $C_i$, there is a group of bosons initially separated from the edge of the cluster by a region of width $L_i$.
Let $A_i^\dag(0) = \prod_{\mathsf{in}_j \in C_i} a_{\mathsf{in}_j}^\dag$ be the creation operator for the group of bosons in the cluster $C_i$.
The initial state is $\ket{\psi(0)} = \prod_{i=1}^K A_i^{\dagger}(0) \ket{0}$.
When evolved for short times, each creation operator $a^\dagger_{\mathsf{in}_i}(t)$ is mostly supported over a small region around its initial position.
Therefore, as long as these regions do not overlap, each operator approximately commutes, and the state is approximately separable.

Let $A^i$ be the smallest ball upon which $A_i^\dagger(0)$ is supported.
Let $B^i_0 = A^i$ and denote its radius $r^i_0$, and define $r_0 = \max r^i_0$.
$B^i_k$ is a ball of radius $r^i_0 + k \ell$ containing $A^i$, where $\ell$ will be chosen to minimize the error.
$S^i_k$ is the shell $B^i_k \setminus B^i_{k-1}$ (see \cref{fig_decompos}).
The complement of a set $X$ is denoted as $X^c$.
We divide the evolution into $N$ time steps between $t_0=0$ and $t_N = t$, and first show that the evolution is well-controlled for one time step from 0 to $t_1 = t/N$.
We apply this decomposition $K$ times, once for each cluster, letting $X = B^i_0$, $Y = S^i_{1}$ and $Z$ be everything else:
\begin{align}
 U_{0,t_1} &\approx U_{0,t_1}^{B^1_1} (U_{0,t_1}^{S^1_1})^\dagger U_{0,t_1}^{(B^1_0)^c} \\
 &\approx U_{0,t_1}^{B^1_1} (U_{0,t_1}^{S^1_1})^\dagger U_{0,t_1}^{B^2_1} (U_{0,t_1}^{S^2_1})^\dagger U_{0,t_1}^{(B^1_0 B^2_0)^c} \\
 &\approx U_{0,t_1}^{B^1_1} (U_{0,t_1}^{S^1_1})^\dagger \ldots U_{0,t_1}^{B^K_1} (U_{0,t_1}^{S^K_1})^\dagger U_{0,t_1}^{(B^1_0 \ldots B^K_0)^c}.
\end{align}
The total error is $O \left( \sum_{i=1}^K (e^{v{t_1}} - 1) \Phi(B^i_0) (\ell^{-\alpha+D+1} + e^{-\gamma \ell}) \right)$ = $O \left(K(e^{v{t_1}} - 1) r_0^{D-1} (\ell^{-\alpha+D+1} + e^{-\gamma \ell}) \right)$.
Applying the decomposed unitary to the initial state and pushing commuting terms through to the vacuum state, we get
\begin{align*}
U_{0,t_1} \ket{\psi(0)} \approx U_{0,t_1}^{B^1_1} A_{1}^\dagger \ldots U_{0,t_1}^{B^K_1} A_{K}^\dagger \ket{0} = \left(\prod_{i=1}^K U_{0,t_1}^{B^i_1} A_{i}^\dagger \right) \ket{0}.
\end{align*}
We can repeat the procedure for the unitary $U_{t_1,t_2}$, where $t_2=2t_1$.
Now, the separating region $Y$ will be $S^i_2$, so that $S^i_2 \cap B^i_1 = \emptyset$.
Each such region still has width $\ell$, but now the boundary of the interior is $\Phi(B^i_1) = O((r_0+\ell)^{D-1})$.
We get
\begin{align}
U_{t_1, t_2} \approx \left(\prod_{i=1}^K U_{t_1, t_2}^{B^i_2}(U_{t_1, t_2}^{S^i_2})^\dagger \right) U_{t_1, t_2}^{(B^1_1 \ldots B^K_1)^c},
\end{align}
with error $O(K(e^{vt_1}-1)(r_0+l)^{D-1}(\ell^{-\alpha+D+1} + e^{-\gamma \ell}))$.
The unitaries supported on $S_2^i$ and $(B_1^1\ldots B_1^K)^c$ commute with all the creation operators supported on sites $B^i_1$, giving 
$\ket{\psi(t_2)} \approx U^{B^1_2}_{t_1, t_2} U^{B^1_1}_{0, t_1} \ldots U^{B^K_2}_{t_1, t_2}U^{B^K_1}_{0, t_1} \ket{\psi(0)}$. 
By applying this procedure a total of $N$ times, once for each time step, we get the approximation $U_{0, t_N}\ket{\psi(0)} \approx U^{B^1_N}_{t_{N-1}, t_N} \ldots U^{B^1_1}_{0, t_1} \ldots U^{B^K_N}_{t_{N-1}, t_N} \ldots U^{B^K_1}_{0, t_1} \ket{\psi(0)}$. 
The total error in the state (in 2-norm) is
\begin{align}
\epsilon &\leq O \left(K(e^{v{t_1}} - 1) (\ell^{-\alpha+D+1} + e^{-\gamma \ell}) \sum_{j=0}^{N-1} (r_0+j\ell)^{D-1}\right) \\
&= O \left( n(e^{v{t_1}} - 1) (\ell^{-\alpha+D+1} + e^{-\gamma \ell}) N L^{D-1} \right), \label{eq_toterr}
\end{align}
proving \cref{lem_easyQHQ}.
The last inequality comes from the fact that $K\leq n$ and that $r_0+(N-1)\ell \leq \min L_i = L$.
The latter condition ensures that the decomposition of the full unitary is separable on the clusters.

In the regime $\alpha $\,$>$\,$ 2D +D/(\beta-1)$, $t_\mathrm{easy}$ is optimized by choosing a fixed time-step size $t_1 $\,$=$\,$ O(1)$.
Then, the number of steps $N$ scales as the evolution time $N$\,$=$\,$t/t_1$.
By the last few time-steps, the bosonic operators have spread out and have a boundary of size $L^{D-1}$, so the boundary terms contribute $O(NL^{D-1})$ in total.
In the regime $D+1 $\,$<$\,$ \alpha $\,$\leq$\,$ 2D +D/(\beta-1)$, the boundary contribution outweighs the suppression $\ell^{-\alpha+D+1}$.
Instead, we use a single time-step in this regime, resulting in $t_\mathrm{easy}=\Omega(\log n)$.

\end{proof}

\section{Section II: \qquad Closeness of evolution under $H$ and $H'$.}
Next, we show that the states evolving due to $H$ and $H'$ are close, owing to the way the truncation works.
This will enable us to prove that the easiness timescale for $H$ is the same as that of $H'$.
Suppose that an initial state $\ket{\psi(0)}$ evolves under two different Hamiltonians $H(t)$ and $H'(t)$ for time $t$, giving the states $\ket{\psi(t)}=U_t\ket{\psi(0)}$ and $\ket{\psi'(t)}=U'_t \ket{\psi(0)}$, respectively.
Define $\ket{\delta(t)} = \ket{\psi(t)} - \ket{\psi'(t)}$ and switch to the rotating frame, $\ket{\delta^r(t)} = U^\dag_t \ket{\delta(t)} = \ket{\psi(0)}- U_t^\dag U'_t \ket{\psi(0)}$.
Now taking the derivative,
\begin{align}
i \partial_t \ket{\delta^r(t)} &= 0 + U_t^\dag H(t) U'_t \ket{\psi(0)} - U_t^\dag H'(t) U'_t \ket{\psi(0)} \\
&= U_t^\dag (H(t) - H'(t)) \ket{\psi'(t)}.
\end{align}
The first line comes about because $i \partial_t U'_t = H'(t) U'_t$ and $i \partial_t U^\dag_t = -U^\dag_t H(t)$, owing to the time-ordered form of $U_t$.

Now, we can bound the norm of the distance, $\delta(t) := \norm{\ket{\delta(t)}} = \norm{\ket{\delta^r(t)}}$.
\begin{align}
\delta(t) &\leq \delta(0) + \int_0^t d\tau \norm{(H(\tau) - H'(\tau)) \ket{\psi'(\tau)}} \\
&= \int_0^t d\tau \norm{(H(\tau) - H'(\tau)) \ket{\psi'(\tau)}} \label{eq_asfaraswecango},
\end{align}
since $\delta(0)=0$.

The next step is to bound the norm of $(H-H')\ket{\psi'(\tau)}$ (we suppress the time label $\tau$ in the argument of $H$ and $H'$ here and below).
We use the HHKL decomposition: $\ket{\psi'(\tau)} = \ket{\phi(\tau)} + \ket{\epsilon(\tau)}$, where the state $\ket{\phi(\tau)}$ is a product state over clusters, and $\ket{\epsilon(\tau)}$ is the error induced by the decomposition.
We first show that $(H-H')\ket{\phi(\tau)} = 0$.
Since $\ket{\phi(\tau)}$ is a product state of clusters, each of which is time-evolved separately, boson number is conserved within each cluster.
Therefore, each cluster has at most $b$ bosons, and $Q\ket{\phi(\tau)} = \ket{\phi(\tau)}$. 
Furthermore, only the hopping terms in $H$ can change the boson number distribution among the different clusters, and these terms move single bosons.
This implies that $H\ket{\phi(\tau)}$ has at most $b+1$ bosons per cluster, and remains within the image of $Q$, denoted $\mathrm{im} \ Q$.
Combining these observations, we get $H'\ket{\phi(\tau)} = QHQ\ket{\phi(\tau)} = H\ket{\phi(\tau)}$.
This enables us to say that $(H-H')\ket{\phi(\tau)} = (H-QHQ)\ket{\phi(\tau)} = 0$.
\Cref{eq_asfaraswecango} gives us 
\begin{align}
\delta(t) &\leq \int_0^t d\tau \norm{(H(\tau) - H'(\tau)) (\ket{\phi(\tau)} + \ket{\epsilon(\tau)})}  \\
&= \int_0^t d\tau \norm{(H(\tau) - H'(\tau)) \ket{\epsilon(\tau)}}, \\
&\leq \max_{\tau, \ket{\eta} \in \mathrm{im\ } Q}\norm{(H(\tau) - H'(\tau))\ket{\eta}} \int_0^t d\tau \norm{\ket{\epsilon(\tau)}}.
\end{align}
In the last inequality, we have upper bounded $\norm{(H(\tau) - H'(\tau)) \ket{\epsilon(\tau)}}$ by $\max_{\ket{\eta} \in \mathrm{im\ } Q}\norm{(H - H')\ket{\eta}} \times {\epsilon(\tau)}$, where $\epsilon(\tau) := \norm{\ket{\epsilon(\tau)}}$.
The quantity $\max_{\ket{\eta} \in \mathrm{im\ } Q}\norm{(H - H')\ket{\eta}} $ can be thought of as an operator norm of $H-H'$, restricted to the image of $Q$.
It is enough to consider a maximization over states $\ket{\eta}$ in the image of $Q$ because we know that the error term $\ket{\epsilon(\tau)}$ also belongs to this subspace, as $\ket{\psi'(\tau)}$ belongs to this subspace.
Further, we give a uniform (time-independent) bound on this operator norm, which accounts for the maximization over times $\tau$.

\begin{lemma}
 $\max_{\ket{\eta}} \norm{(H-QHQ)\ket{\eta}} \leq \norm{\sum_{i \in C_k,j \in C_l} J_{ij}a_i^\dag a_j} \leq O(b L^{D-\alpha})$.
\end{lemma}
\begin{proof}
Notice that for each term $H_i$ in the Hamiltonian, the operator $H-QHQ$ contains $H_i - Q H_i Q$, where the rightmost $Q$ can be neglected since $Q \ket{\eta} = \ket{\eta}$.
The on-site terms $\sum_i J_{ii}a^\dag_i a_i + V n_i (n_i-1)/2 $ do not change the boson number. Therefore, they cannot take $\ket{\eta}$ outside the image of $Q$, and do not contribute to $(H - QHQ)\ket{\eta}$.
The only contribution comes from hopping terms that change boson number, which we bound by 
\begin{align}
\norm{\sum_{i \in C_k,j \in C_l} J_{ij}a_i^\dag a_j},
\end{align}
where the sum is over sites $i$ and $j$ in distinct clusters $C_k$ and $C_l$, respectively.
This is because only hopping terms that connect different clusters can bring $\ket{\eta}$ outside the image of $Q$, since hopping terms within a single cluster maintain the number of bosons per cluster.

For illustration, let us focus on terms that couple two clusters $C_1$ and $C_2$.
The distance between these two clusters is denoted $L_{12}$. For any coupling $J_{ij}$ with $i \in C_1$ and $j \in C_2$, we can bound $\abs{J_{ij}} \leq L_{12}^{-\alpha}$ by assumption.
Let 
\begin{align}
H^{\mathrm{hop}}_{12} = \sum_{i \in C_1 \ j \in C_2} J_{ij} a_i^\dag a_j + \mathrm{h.c.}                                                                                        \end{align}
denote the sum over all such pairs of sites.
Then, we can bound $\norm{H^{\mathrm{hop}}_{12}\ket{\eta}} \leq O(b)$.
To see this, diagonalize $H^{\mathrm{hop}}_{12} = \sum_{i} w_i b_i^\dag b_i$. Since $H^{\mathrm{hop}}_{12}$ only acts on two clusters, each normal mode contains up to $2b$ bosons.
The maximum eigenvalue of $H^\mathrm{hop}_{12}$ is bounded by $2b\max_i w_i$, where $w_i$ is the maximum normal mode frequency, given by the eigenvalue of the matrix $J_{ij}: i\in C_1, j \in C_2$.
We now apply the Gershgorin circle theorem, which states that the maximum eigenvalue of ${J}$ is bounded by the quantity $\max_i (\sum_j \abs{J_{ij}})\leq L^D L_{12}^{-\alpha} $.

Taking advantage of the fact that the clusters form a cubic lattice in $D$ dimensions, we can group pairs of clusters by their relative distances.
If we label clusters $i$ by their $D$-dimensional coordinate $i_1,i_2,...,i_D$, then we can define the cluster distance $l$ between $i$ and $j$ as $l+1= \max_d \abs{i_d - j_d}$.
Cluster distance $l$ corresponds to a minimum separation $l \times L$  between sites in different clusters.
With this definition, there are $((2l+3)^D - (2l+1)^D) \approx 2^D D l^{D-1}$ clusters at a cluster distance $l$ from any given cluster (\cref{fig_clusterpairing}(a)), and $K \times 2^D D l^{D-1}$ total pairs of clusters at cluster distance $l$.
Notice that for a given separation vector, $K/2$ pairs of clusters ($K$ total) can be simultaneously coupled without overlap (\cref{fig_clusterpairing}(b)).
Therefore, there are approximately $2^{D+1} D l^{D-1}$ non-overlapping groupings per distance $l$.
The sum over these non-overlapping Hamiltonians $H^{\mathrm{hop}}_{a_1 b_1}+...+H^{\mathrm{hop}}_{a_{K/2} b_{K/2}}$ for each grouping is block diagonal. Therefore, the spectral norm (maximum eigenvalue) of the total Hamiltonian is equal to the maximum of the spectral norm over all irreducible blocks.
Putting all this together, as long as $D-1-\alpha < -1$ the bound becomes 
\begin{align}
\max_{\ket{\eta} \in \mathrm{im\ } Q} \norm{(H - QHQ)\ket{\eta}} \leq \sum_{l=0}^{l_{\mathrm{max}}}  O(2^{D+1} D l^{D-1}) (2 b) L^D (l L)^{-\alpha} = O(b L^{D-\alpha}).
\end{align}
\end{proof}

\begin{figure}[]
\includegraphics[width=0.8\linewidth]{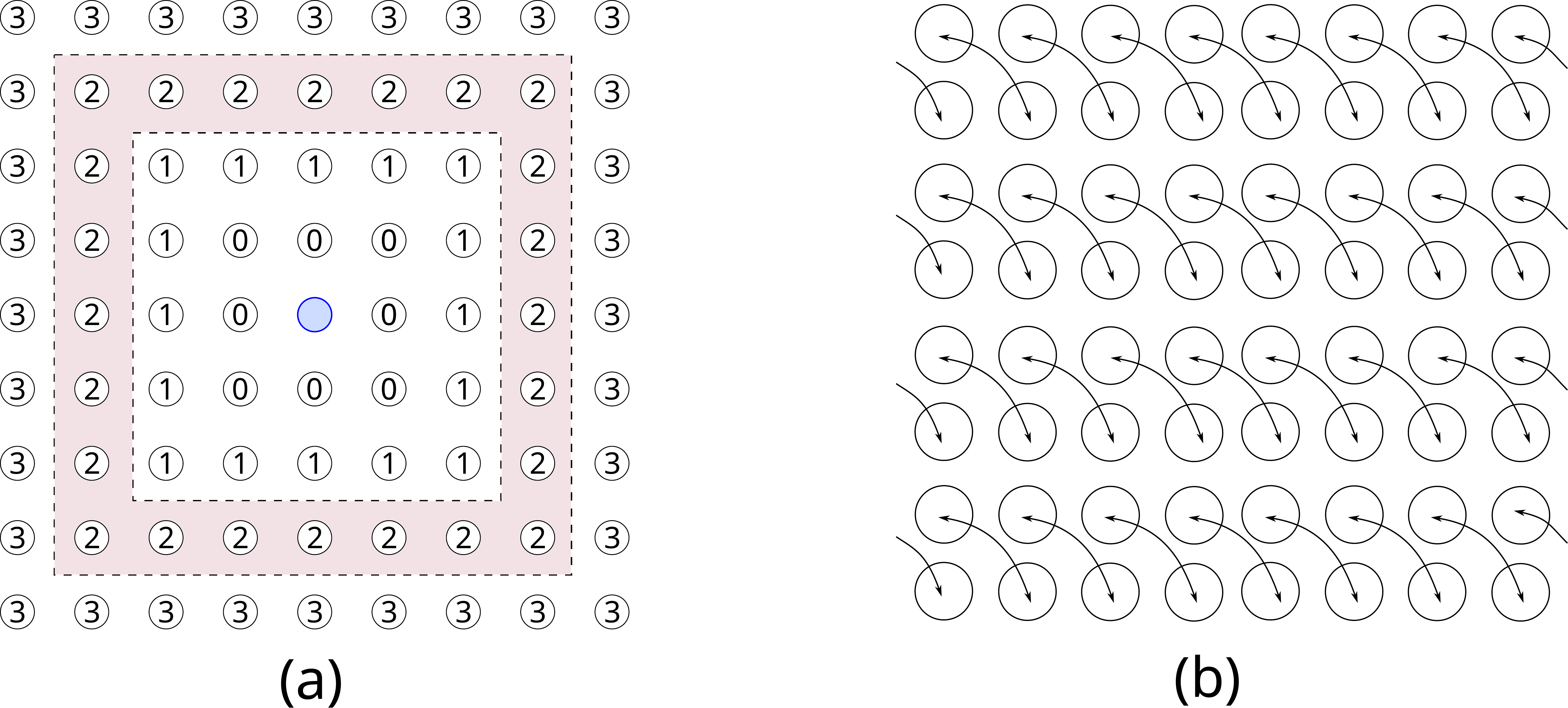}
\caption{(a) Cluster distances between the blue cluster in the center and nearby clusters.
The total number of clusters at cluster distance $l=2$ (pink background) is given by $(2l+3)^2 - (2l+1)^2  = 24$. (b) Non-overlapping pairing between clusters separated by a diagonal.
The distance between these clusters is $l=0$, since they share a boundary and contain adjacent sites.} \label{fig_clusterpairing}
\end{figure}

We are now in a position to prove \cref{thm_easy}.
\begin{proof}[Proof of \cref{thm_easy}]
There are two error contributions, $\epsilon$ and $\delta$, to the total error.
The HHKL error $\epsilon$ is given by evaluation of \cref{eq_toterr}, which is minimized by either choosing $N=1$ or $N = t/t_1$ with $t_1$ a small fixed constant. This leads to three regimes with errors
\begin{align}
    \epsilon \leq O(1) \times 
    \begin{cases}
    ne^{vt-L} ,\ \ & \alpha \rightarrow \infty \\
    \frac{n t^{\alpha-D}}{L^{\alpha-2D}} ,\ \ & 2D +\frac{D}{\beta -1} < \alpha < \infty \\
    \frac{n (e^{vt}-1)}{L^{\alpha-D-1}},\ \ & D+1 < \alpha \leq 2D + \frac{D}{\beta - 1}.
    \end{cases}
    \label{eq_epsilonerrs}
\end{align}

The truncation error, arising from using $H'$ rather than $H$ in the first step, is given by
\begin{align}
{\delta(t)} &\leq O(b L^{D-\alpha}) \int_0^t d\tau \epsilon(\tau).
\end{align}
Therefore, we can upper bound $\delta(t)$ by $\epsilon$ times an additional factor.
This factor is $bL^{D-\alpha} t$, $\delta(t) = b L^{D-\alpha} t \epsilon$, when $\epsilon(\tau) = \poly(\tau)$ ($\alpha < \infty$), and it is $L^{D-\alpha}$, $\delta(t) = L^{D-\alpha} \epsilon$ when $\epsilon(\tau) = \exp(v\tau)$ ($\alpha \rightarrow \infty$).
Our easiness results only hold for $\alpha > D+1$, so the $L$-dependent factor serves to suppress the truncation error in the asymptotic limit.
Although the additional factor of $t$ could cause $\delta(t) > \epsilon$ at late times, by this time, $\epsilon > \Omega(1)$ and we are no longer in the easy regime.
Therefore, the errors presented in \cref{eq_epsilonerrs} can be immediately applied to calculate the timescales from the main text.
\end{proof}
The resulting timescales are summarized in \cref{tab:easiness timescales}, which highlights the scaling of the timescale with respect to different physical parameters.
We also consider the scaling of the easiness timescales when the density of the bosons increases by a factor $\rho$.
In our setting, we implement this by scaling the number of bosons by $\rho$ while keeping the number of lattice sites and the number of clusters (and their size) fixed.
The effect of this is to increase the Lieb-Robinson velocity: $v \rightarrow v \rho$.
For all three cases, the net effect of increasing the density  by a factor $\rho$ is to decrease the easiness timescale.

\begin{center}
\begin{table}[]
\centering
\begin{tabular}{@{}llll@{}} \toprule
Regime & Error & $t_\mathrm{easy}(n,L)$ & $t_\mathrm{easy}(\rho)/t_\mathrm{easy}(\rho=1)$ \\ \midrule
$\alpha \rightarrow \infty$ & $n e^{vt-L}$ & $L$ & $1/\rho$ \\
$2D + \frac{D}{\beta - 1}< \alpha < \infty$ & $\frac{n vt^{\alpha - D}}{L^{\alpha-2D}}$ & $n^{\frac{-1}{\alpha-D}}L^{\frac{\alpha-2D}{\alpha-D}}$ & $\rho^{\frac{-2}{\alpha -D}}$ \\
$D+1 < \alpha  \leq 2D + \frac{D}{\beta - 1}$\qquad & $\frac{n (e^{vt}-1)}{L^{\alpha-D-1}}$ \qquad & $({\alpha-D-1})\log L - \log n$ \qquad& $1/\rho$ \\ \bottomrule
\end{tabular}
\caption{Summary of easiness timescales in the different regimes. Timescales follow from the error and are presented first as a function of $n$ and $L$, which are the relevant physical scales of the problem. We study the effect of the density by performing the scaling $n\to \rho n, L \to L, b \to \rho b, m \to m$. The last column shows the timescale as a function of $\rho$ in terms of the timescale when $\rho=1$, namely $t_\mathrm{easy}(\rho=1)$.}
\label{tab:easiness timescales}
\end{table} 
\end{center}

\section{Section III: \qquad Hardness timescale for interacting bosons}
In this section we provide more details about how to achieve the timescales in \cref{thm_hard}.
In the interacting case, almost any interaction is universal for \BQP\ \cite{Childs2013} and hence these results are applicable to general on-site interactions $f(n_i)$.

We first describe how a bosonic system with fully controllable local fields $J_{ii}(t)$, hoppings $J_{ij}(t)$, and a fixed Hubbard interaction $\frac{V}{2} \sum_i \hat{n}_i (\hat{n}_i - 1)$ can implement a universal quantum gate set. 
To simulate quantum circuits, which act on two-state spins, we use a dual-rail encoding.
Using $2n$ bosonic modes, and $n$ bosons, $n$ logical qubits are defined by partitioning the lattice into pairs of adjacent modes, and a boson is placed in each pair.
Each logical qubit spans a subspace of the two-mode Hilbert space.
Specifically, $\ket{0}_L = \ket{10}, \ket{1}_L = \ket{01}$.
We can implement any single qubit (2-mode) unitary by turning on a hopping between the two sites ($X$-rotations) or by applying a local on-site field ($Z$-rotations).
To complete a universal gate set, we need a two-qubit entangling gate.
This can be done, say, by applying a hopping term between two sites that belong to different logical qubits \cite{Underwood2012}.
All these gates are achievable in $O(1)$ time when $V = \Theta(1)$.
In the limit of large Hubbard interaction $V \to \infty$, the entangling power of the gate decreases as $1/V$ \cite{Underwood2012} and one needs $O(V)$ repetitions of the gate in order to implement a standard entangling gate such as the CNOT.

For hardness proofs that employ postselection gadgets, we must ensure that the gate set we work with comes equipped with a Solovay-Kitaev theorem.
This is the case if the gate set is closed under inverse, or contains an irreducible representation of a non-Abelian group \cite{Bouland2018b}.
In our case, the gate set contains single-qubit Paulis and hence has a Solovay-Kitaev theorem, which is important for the postselection gadgets to work as intended.

\begin{figure}[]
 \includegraphics[width=0.5\linewidth]{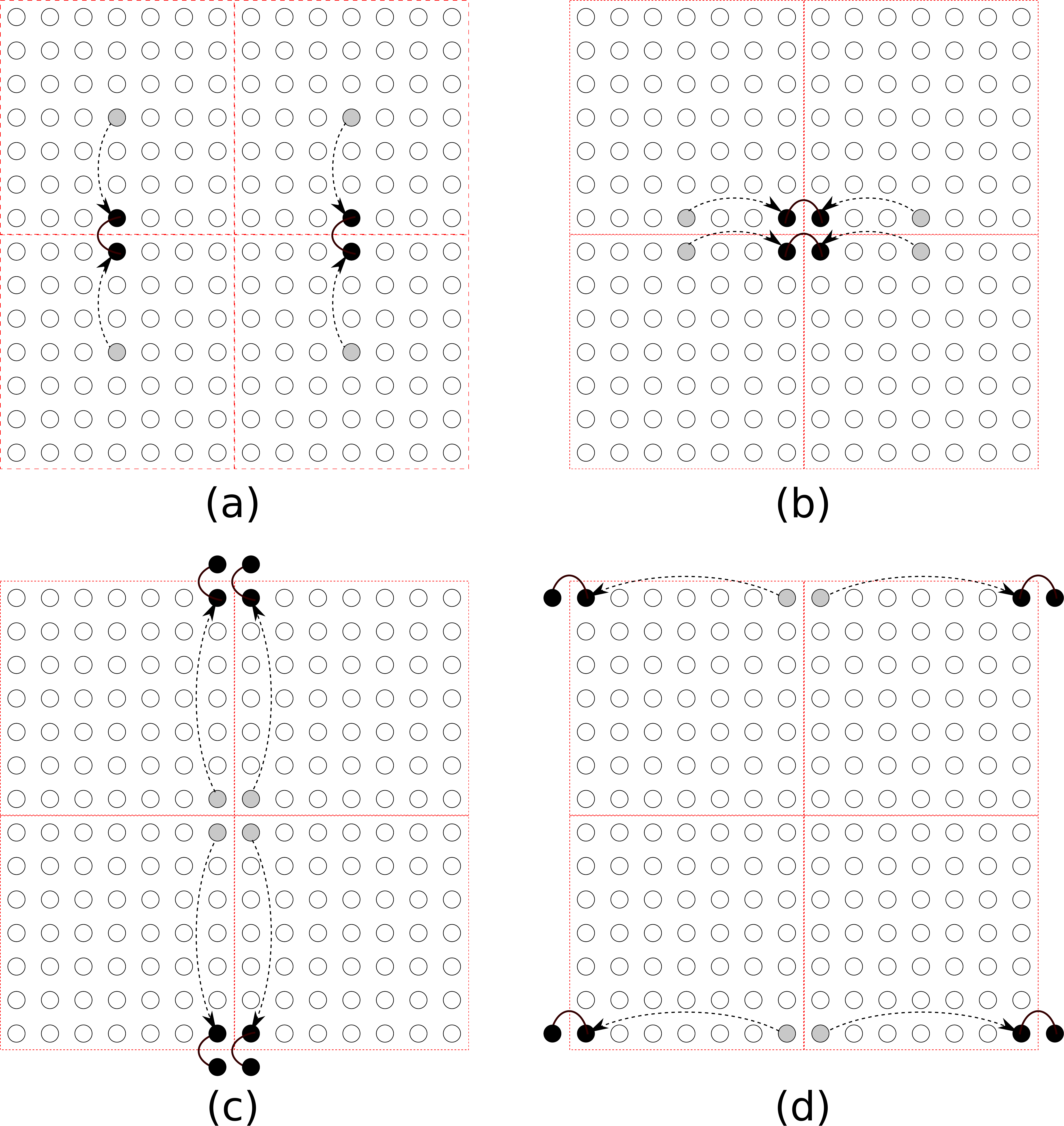}
 \caption{A protocol that implements the logical circuit of Ref.~\cite{Bermejo-Vega2018}.
 Each subfigure shows the location of the site that previously encoded the $\ket{1}$ state in gray.
 The current site that encodes the $\ket{1}$ state is in black.
 The site that encodes $\ket{0}$ is not shown but moves similarly as the $\ket{1}$ state.
 The distance traversed by each qubit is $L + L + 2L + 2L = 6L$.} \label{fig_scheme}
\end{figure}

We will specifically deal with the scheme proposed in Ref.~\cite{Bermejo-Vega2018}.
It applies a constant-depth circuit on a grid of $\sqrt{n}\times \sqrt{n}$ qubits in order to implement a random IQP circuit \cite{Bremner2011,Bremner2016} on $\sqrt{n}$ effective qubits.
This comes about because the cluster state, which is a universal resource for measurement-based quantum computation, can be made with constant depth on a two-dimensional grid.

For short-range hops ($\alpha \rightarrow \infty$), we implement the scheme in four steps as shown in \cref{fig_scheme}.
In each step, we move the logical qubits to bring them near each other and make them interact in order to effect an entangling gate.
For short-range hopping, the time taken to move a boson to a far-off site distance $L$ away dominates the time taken for an entangling gate.
The total time for an entangling gate is thus $O(L) + O(1) = O(L)$.

For long-range hopping, we use the same scheme as in \cref{fig_scheme}, but we use the long-range hopping to speed up the movement of the logical qubits.
This is precisely the question of state transfer using long-range interactions/hops \cite{Eldredge2017,Guo2019a,Tran2020}.
In the following we give an overview of the best known protocol for state transfer, but first we should clarify the assumptions in the model.
The Hamiltonian is a sum of $O(m^2)$ terms, each of which has norm bounded by at most $1/d(i,j)^\alpha$.
Since we assume we can apply any Hamiltonian subject to these constraints, in particular, we may choose to apply hopping terms across all possible edges.
This model makes it possible to go faster than the circuit model if we compare the time in the Hamiltonian model with depth in the circuit model.
This power comes about because of the possibility of allowing simultaneous noncommuting terms to be applied in the Hamiltonian model.

The state transfer protocols in Ref.~\cite{Guo2019a,Tran2020} show such a speedup for state transfer.
The broad idea in both protocols is to apply a map $\ket{1}_1 \rightarrow \ket{1}_A := \sum_{j \in A} \frac{1}{\sqrt{\abs{A}}} \ket{1}_j$, followed by the steps $\ket{1}_A \rightarrow \ket{1}_B$ and $\ket{1}_B \rightarrow \ket{1}_2$, where $A$ and $B$ are regions of the lattice to be specified.
In the protocol of Ref.~\cite{Guo2019a}, which is faster than that of Ref.~\cite{Tran2020} for $\alpha \leq D/2$, $A = B = \{j: j \neq 1,2\}$ and each step takes time $O(L^\alpha/\sqrt{N-2})$, where $N-2$ is the number of ancillas used and $L$ is the distance between the two furthest sites.
In the protocol of Ref.~\cite{Tran2020}, which is faster for $\alpha \in (D/2,D+1]$, $A$ and $B$ are large regions around the initial and final sites, respectively.
This protocol takes time $O(1)$ when $\alpha < D$, $O(\log L)$ when $\alpha = D$, and $O(L^{\alpha - D})$ when $\alpha > D$.

In our setting, we use the state transfer protocols to move the logical qubit faster than time $O(L)$ in each step of the scheme depicted in \cref{fig_scheme}.
If $\alpha <D/2$, we use all the ancillas in the entire system, giving a state transfer time of $O(m^{\alpha/D -1/2}) = O(n^{\beta (\frac{\alpha}{D} - \frac{1}{2})})$.
If $\alpha > D/2$, we only use the empty sites in a cluster as ancillas in the protocol of Ref.\,\cite{Tran2020}, giving the state transfer time mentioned above.
This time is faster than $O(L)$, the time it would take for the nearest-neighbor case, when $\alpha < D + 1$.
Therefore, for 2D or higher and $\alpha \geq D/2$, the total time it takes to implement a hard-to-simulate circuit is $\min[L, L^{\alpha -D}\log L]+O(1)$, proving \cref{thm_hard} for interacting bosons.
When $\alpha < D/2$, the limiting step is dominated by the entangling gate, which takes time $O(1)$.
Therefore for this case we only get fast hardness through boson sampling, which is discussed in Section IV.
Note that when $t=o(1)$ and interaction strength is $V=\Theta(1)$, the effect of the interaction is governed by $Vt = o(1)$, which justifies treating the problem for short times as a free-boson problem.

\subsection{III.A: \qquad One dimension}
In 1D with nearest-neighbor hopping, we cannot hope to get a hardness result for simulating constant depth circuits, which is related to the fact that one cannot have universal measurement-based quantum computing in one dimension.
We change our strategy here.
The overall goal in 1D is to still be able to simulate the scheme in Ref.~\cite{Bermejo-Vega2018} since it provides a faster hardness time (at the cost of an overhead in the qubits).
The way this is done is to either (i) implement $O(n)$ SWAPs in 1D in order to implement an IQP circuit \cite{Bremner2011}, or (ii) use the long-range hops to directly implement gates between logical qubits at a distance $L$ away.

For the first method, we use state transfer to implement a SWAP by moving each boson within a cluster a distance $\Theta(L)$.
This takes time $O(t_s(L))$, where $t_s(L)$ is the time taken for state transfer over a distance $L$ and is given by
\begin{align}
t_s(L) = c \times
\begin{cases}
 L, & \alpha > 2 
 \\ L^{\alpha-1}, & \alpha \in (1,2]
 \\ \log L, & \alpha = 1
 \\ 1, & \alpha \in [\frac{1}{2},1)
 \\ L^{\alpha -1/2}, & \alpha < \frac{1}{2}.
\end{cases}
\end{align}
We write this succinctly as $O\left(\min[L,L^{\alpha - 1}\log L +1, L^{\alpha -1/2}]\right)$.
The total time for $n$ SWAPs is therefore $O\left(n\times \min[L,L^{\alpha - 1}\log L +1, L^{\alpha -1/2}]\right)$.
\begin{figure} \centering
 \includegraphics[width=0.45\linewidth]{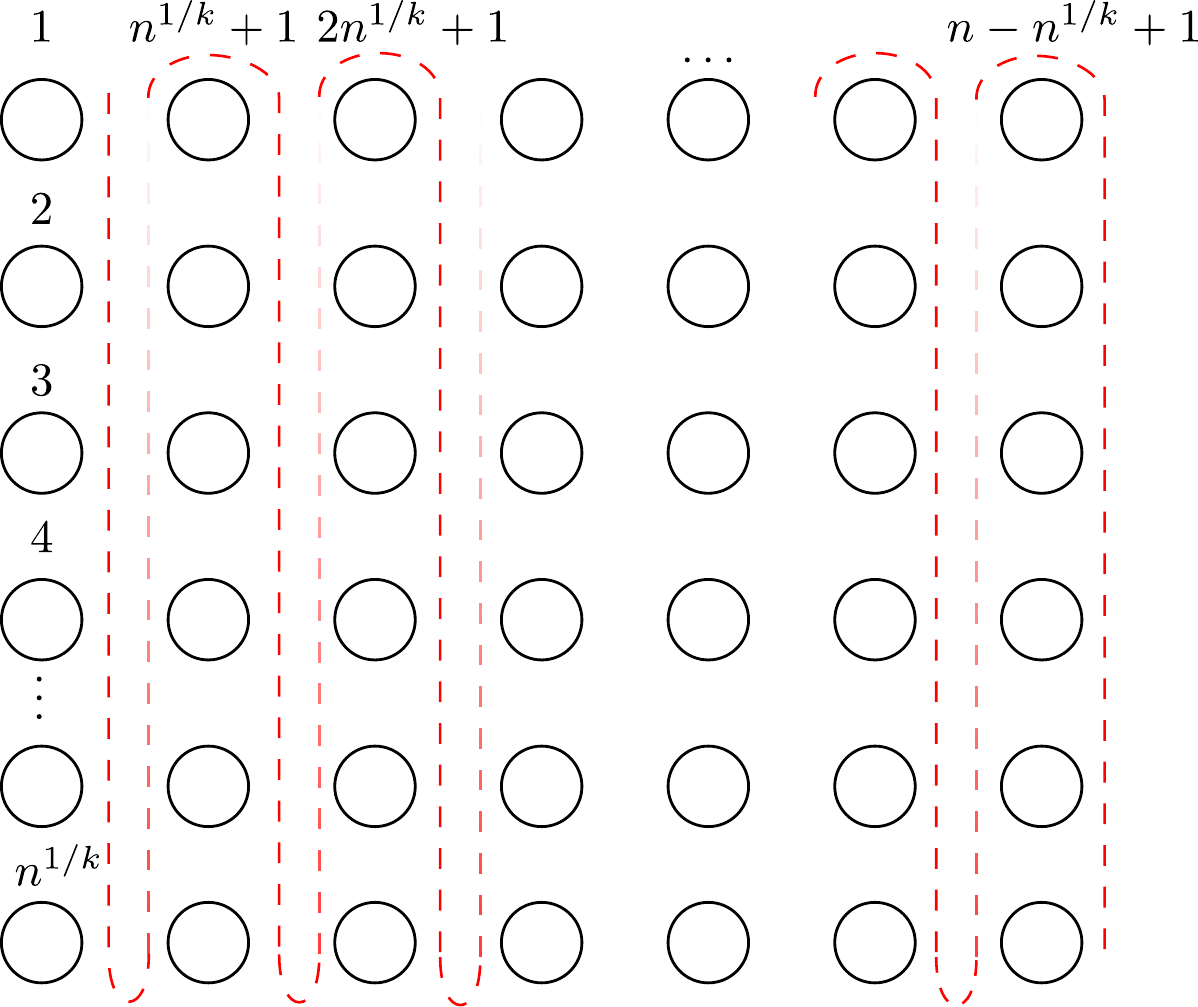}
 \caption{A snaking scheme to assign indices to qubits in 2D for a $n^{1/k} \times n^{1-1/k}$ grid, which is used in mapping to 1D.}\label{fig_snake}
\end{figure}

The second method relies on the observation that when $\alpha \rightarrow 0$, the distinction between 1D and 2D becomes less clear, since at $\alpha=0$, the connectivity is described by a complete graph and all hopping strengths are equal.
Let us give some intuition for the $\alpha \rightarrow 0$ case.
One would directly ``sculpt'' a 2D grid from the available graph, which is a complete graph on $n$ vertices (one for every logical qubit) with weights $w_{ij}$ given by $d(i,j)^{-\alpha}$.
If we want to arrange qubits on a 1D path, we can assign an indexing to qubits in the 2D grid and place them in the 1D path in increasing order of their index.
One may, in particular, choose a ``snake-like indexing'' depicted in \cref{fig_snake}.
This ensures that nearest-neighbor gates along one axis of the 2D grid map to nearest-neighbor gates in 1D.
Gates along the other axis, however, correspond to nonlocal gates in 1D.
Suppose that the equivalent grid in 2D is of size $n^{1/k} \times n^{1-1/k}$.
The distance between two qubits that have to participate in a gate is now marginally larger ($O(Ln^{1/k})$ instead of $O(L)$), but the depth is greatly reduced: it is now $O(n^{1/k})$ instead of $O(n)$.
We again use state transfer to move close to a far-off qubit and then perform a nearest-neighbor entangling gate.
This time is set by the state transfer protocol, and is now $t_s(n^{1/k}L) = O\left(n^{1/k}\times \min[L,L^{\alpha - 1}\log L+1, L^{\alpha -1/2}]\right)$.
For large $k=\Theta(1)$, this gives us the bound $O\left(\min[L^{1+\delta}, L^{\alpha - 1 + \delta}+L^{\Theta(\delta)},L^{\alpha -1/2 + \delta}]\right)$ for any $\delta > 0$, giving a coarse transition.
Notice, however, that faster hardness in 1D comes at a high cost-- the effective number of qubits on which we implement a hard circuit is only $\Theta(n^{1/k}) = n^{\Theta(\delta)}$, which approaches a constant as $\delta \rightarrow 0$.

This example of 1D is very instructive-- it exhibits one particular way in which the complexity phase transition can happen.
As we take higher and higher values of $k$, the hardness time would decrease, coming at the cost of a decreased number of effective qubits.
This smoothly morphs into the easiness regime when $\alpha \rightarrow \infty$ since in this regime both transitions happen at $t = \Theta(L)$.

If the definition of hardness is more stringent (in order to link it to fine-grained complexity measures such as explicit quantitative lower bound conjectures), then the above mentioned overhead is undesirable.
In this case we would adopt the first strategy to implement SWAPs and directly implement a random IQP circuit on all the $n$ qubits.
This would increase the hardness time by a factor $n$.

\subsection{III.B: \qquad Hardcore limit}
\begin{figure}
\centering
\includegraphics[width=0.6\linewidth]{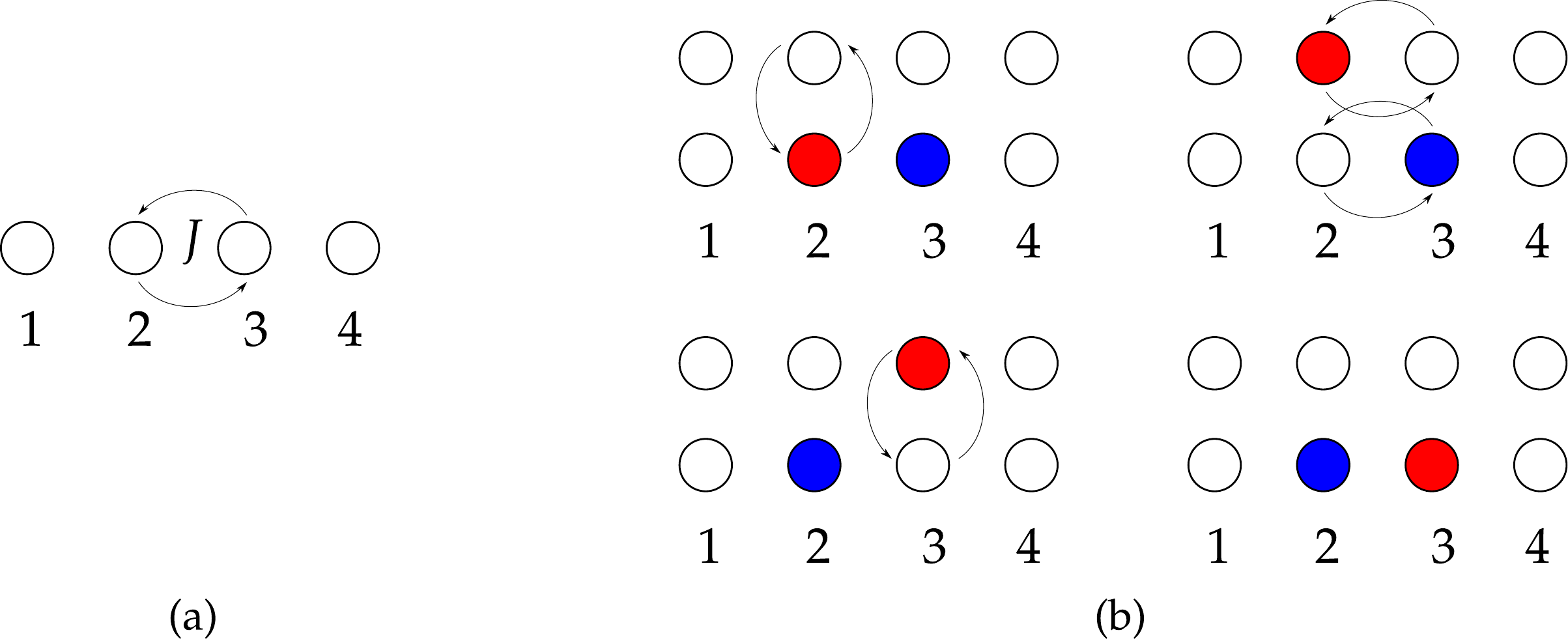}
\caption{(a) A hopping between sites 2 and 3 that implements the mode unitary 
$\begin{pmatrix}
\cos(\abs{J}t) & -i\sin(\abs{J}t){J/\abs{J}} \\
-i\sin(Jt){J^*/\abs{J}} & \cos(\abs{J}t) 
\end{pmatrix} = e^{-it (\Re{J} X + \Im{J}Y)}$.
When $\abs{J}t  = \pi$, this is a SWAP between two modes with phases $(-i{J/\abs{J}},-i{J^*/\abs{J}})$ that depends on $\arg{J}$, the argument of $J$.
(b) A ``physical'' SWAP between sites 2 and 3 by using ancilla sites available whenever the system is not nearest-neighbor in 1D.
The colors are used to label the modes and how they move, and do not mean that both sites are occupied.
The total hopping phase incurred when performing the physical SWAP can be set to be $(+i,-i)$, which cannot be achieved with just the hopping term shown in (a).} \label{fig_hcdualrail}
\end{figure}

In the hardcore limit $V \rightarrow \infty$, the strategy is modified.
Let us consider a physical qubit to represent the presence ($\ket{1}$) or absence ($\ket{0}$) of a boson at a site.
A nearest-neighbor hop translates to a term in the Hamiltonian that can be written in terms of the Pauli operators as $XX + YY$.
Further, an on-site field $J_{ii} a_i^\dag a_i$ translates to a term $\propto Z$.
There are no other terms available, in particular single-qubit rotations about other axes $X$ or $Y$.
This is because the total boson number is conserved, which in the spin basis corresponds to the conservation of $\sum_i Z_i$.
This operator indeed commutes with both the allowed Hamiltonian terms specified above.

Let us now discuss the computational power of this model.
When the physical qubits are constrained to have nearest-neighbor interactions in 1D, this model is nonuniversal and classically simulable.
This can be interpreted due to the fact that this model is equivalent to matchgates on a path (i.e.\ a 1D nearest-neighbor graph), which is nonuniversal for quantum computing without access to a SWAP gate.
Alternatively, one can apply the Jordan-Wigner transformation to map the spin model onto free fermions.
One may then use the fact that fermion sampling is simulable on a classical computer \cite{Terhal2002}.

When the connectivity of the qubit interactions is different, the model is computationally universal for \BQP.
In the matchgate picture, this result follows from Ref.~\cite{Brod2013}, which shows that matchgates on any graph apart from the path or the cycle are universal for \BQP \ in an encoded sense.
In the fermion picture, the Jordan-Wigner transformation on any graph other than a path graph would typically result in nonlocal interacting terms that are not quadratic in general.
Thus, the model cannot be mapped to free, quadratic fermions and the simulability proof from Ref.~\cite{Terhal2002} breaks down.

Alternatively, a constructive way of seeing how we can recover universality is as follows.
Consider again the dual rail encoding and two logical qubits placed next to each other as in \cref{fig_hcdualrail}.
Apply a coupling $J(a_2^\dag a_3 + a_3^\dag a_2)$ on the modes 2 and 3 for time $t=\frac{\pi}{2J}$.
This effects the transition $\ket{10}_{23} \rightarrow -i \ket{01}_{23}$ and $\ket{01}_{23} \rightarrow -i \ket{10}_{23}$, while leaving the state $\ket{11}_{23}$ the same.
Now we swap the modes 2 and 3 using an ancilla mode that is available by virtue of having either long-range hopping or having $D>1$.
This returns the system back to the logical subspace of exactly one boson in modes 1 \& 2, and one boson in modes 3 \& 4, and effects the unitary $\mathrm{diag}\{1,1,1,-1\}$ in the (logical) computational basis.
This is an entangling gate that can be implemented in $O(1)$ time and thus the hardness timescale for hardcore interactions is the same as that of Hubbard interactions with $V=\Theta(1)$.

We finally discuss the case when $V$ is polynomially large.
Using the dual-rail encoding and implementing the same protocol as the non-hardcore case now takes the state $\ket{11}_{23}$ to $\lambda \ket{11}_{23} + \mu \frac{\ket{20}_{23}+\ket{02}_{23}}{\sqrt{2}}$, with $\mu \propto\frac{J}{\sqrt{8J^2 + V^2}} \sin \left(\frac{t \sqrt{8 J^2+V^2}}{2}\right)$.
When $\abs{\mu} \neq 0$, we get an error because the state is outside the logical subspace.
The probability with which this action happens is suppressed by $1/V^2$, however, which is polynomially small when $V=\poly(n)$.

However, one can do better: by carefully tuning the hopping strength $J \in [0,1]$ and the evolution time $t$, one can always achieve the goal of getting $\mu = 0$ exactly and implementing an operation $\exp[-i \frac{\pi}{2}X]$ in the $\ket{10}_{23},\ket{01}_{23}$ subspace.
This requires setting $t\sqrt{2J^2 + \frac{V^2}{4}} = m\pi$ and $t = \frac{2\pi}{J}$ for integer $m$.
This can be solved as follows: set $m = \lceil \sqrt{8+V^2} \rceil$, and $J = \frac{V}{\sqrt{m^2 -8}}$ (which is $\leq 1$ since $ m \geq \sqrt{8+V^2}$).
The time is set by the condition $t  = \frac{2\pi}{J}$, which is $\Theta(1)$.
This effects a logical CPHASE$[\phi]$ gate with angle $\phi = -\pi V/J$.

Finally, the above parameters that set $\mu$ exactly to zero work even for exponentially large $V = \Omega(\exp(n))$, but this requires exponentially precise control of the parameters $J$ and $t$, which may not be physically feasible.
In this case, we simply observe that $\abs{\mu}^2$, the probability of going outside the logical subspace and hence making an error, is $O(1/V^2)$, which is exponentially small in $n$.
Therefore, in this limit, the gate we implement is exponentially close to perfect, and the complete circuit has a very small infidelity as well.

\section{Section IV: \qquad Hardness timescale for free bosons} \label{sec_freebosonhard}
In this section, we review Aaronson and Arkhipov's method of creating a linear optical state that is hard to sample from \cite{Aaronson2011}.
We then give a way to construct such states in time $\tilde{O}\left(n{{m^{\alpha/D - 1/2}}}\right)$ with high probability in the Hamiltonian model, and prove \cref{thm_hard} for free bosons.

For free bosons, in order to get a state that is hard to sample from, we need to apply a Haar-random linear-optical unitary on $m$ modes to the state $\ket{1,1,\ldots 1,0,0,\ldots0}$.
Aaronson and Arkhipov gave a method of preparing the resulting state in $O(n\log m)$ depth in the circuit model.
Their method involves the use of ancillas and can be thought of as implementing each column of the Haar-random unitary separately in $O(\log m)$-depth.
Here we mean that we apply the map $\ket{1}_j \rightarrow \sum_{i \in \Lambda} U_{ij} \ket{1}_i $ to ``implement" the column $i$ of the linear-optical unitary $U$.
In the Hamiltonian model, we can apply simultaneous, non-commuting terms of a Hamiltonian involving a common site.
The only constraint is that each term of the Hamiltonian should have a bounded norm of $1/d(i,j)^\alpha$.
In this model, when $\alpha$ is small, it is possible to implement each unitary in a time much smaller than $O(\log m)$-- indeed, we show the following:
\begin{lemma} \label{lem_fastcol}
 Let $U$ be a Haar-random unitary on $m$ modes.
 Then with probability $1- \frac{1}{\poly(m)}$ over the Haar measure, each of the first $n$ columns of $U$ can be implemented in time $O\left({\frac{\sqrt{\log m}}{m^{1/2 - \alpha/D}}}\right)$.
\end{lemma}
To prove this, we will need an algorithm that implements columns of the unitary.
For convenience, let us first consider the case $\alpha = 0$.
The algorithm involves two subroutines, which we call the single-shot and state-transfer protocols.
Both protocols depend on the following observation.
If we implement a Hamiltonian that couples a site $i$ to all other sites $j\neq i$ through coupling strengths $J_{ij}$, then the effective dynamics is that of two coupled modes $a_i^\dagger$ and $b^\dagger = \frac{1}{\omega}\sum_{j\neq i} J_{ij} a_j^\dagger$, where $\omega = \sqrt{\sum_{j\neq i} J_{ij}^2 }$.
The effective speed of the dynamics is given by $\omega$-- for instance, the time period of the system is $\frac{2\pi}{\omega}$.

The single-shot protocol implements a map $a_i^\dagger \rightarrow \gamma_i a_i^\dagger + \sum_{j\neq i} \gamma_j a_j^\dagger$.
This is done by simply applying the Hamiltonian $H \propto a_i^\dagger (\sum_{j\neq i} \gamma_j a_j) + \mathrm{h.c.\ }$ for time $t=\frac{1}{\omega}\cos^{-1}\abs{\gamma_i}$.
In the case $\alpha=0$, we can set the proportionality factor equal to $1/\mathrm{max} |\gamma_j|$.
This choice means that the coupling strength between $i$ and the site $k$ with maximum $|\gamma_k|$ is set to 1 (the maximum), and all other couplings are equal to $|\frac{\gamma_j}{\gamma_k}|$.

The other subroutine, the state-transfer protocol is also an application of the above observation and appears in Ref.~\cite{Guo2019a}.
It achieves the map $a_i^\dagger \rightarrow \gamma_ia_i^\dagger + \gamma_j a_j^\dagger$ via two rounds of the previous protocol.
This is done by first mapping site $i$ to the uniform superposition over all sites except $i$ and $j$, and then coupling this uniform superposition mode to site $j$.
The time taken for this is $\frac{1}{\omega}\left( \frac{\pi}{2} + \cos^{-1}\abs{\gamma_i} \right)$.
Since $\omega = \sqrt{m-2}$ (all $m-2$ modes are coupled with equal strength to modes $i$ or $j$), this takes time $O\left( \frac{1}{\sqrt{m}} \right)$.

These subroutines form part of \cref{alg_1col}.
\begin{algorithm}[h]
 \caption{Algorithm for implementing one column of a unitary} \label{alg_1col}
 \SetInd{1em}{1em}
\DontPrintSemicolon
 \KwIn{Unitary $U$, column index $j$}
\ Reassign the mode labels for modes $i\neq j$ in nonincreasing order of $|U_{ij}|$. \;
\ Implement the state-transfer protocol to map the state $a_j^\dagger \ket{\mathrm{vac}}$ to $U_{jj} a_j^\dagger \ket{\mathrm{vac}} + \sqrt{1-|U_{jj}|^2} a_1^\dagger \ket{\mathrm{vac}} $. 
Skip this step if $|U_{jj}| \geq |U_{j1}|$ already. \;
\ Use the single-shot protocol between site $1$ and the rest ($i\neq 1,j$) to map $a_1^\dagger \rightarrow \frac{U_{1j}}{\sqrt{1-|U_{jj}|^2}} a_1^\dagger + \sum_{i\neq1,j} \frac{U_{ij}}{\sqrt{1-|U_{jj}|^2}} a_i^\dagger$.
\end{algorithm}
It can be seen that \cref{alg_1col} implements a map $a_j^\dagger \rightarrow U_{jj}a_j^\dagger + \sum_{i\neq j} U_{ij} a_i^\dagger$, as desired.
To prove \cref{lem_fastcol} we need to examine the runtime of the algorithm when $U$ is drawn from a Haar-random distribution.
\begin{proof}[Proof of \cref{lem_fastcol}]
First, notice that since the Haar measure is invariant under the action of any unitary, we can in particular apply a permutation map to argue that the elements of the $i$'th column are drawn from the same distribution as the first column.
Next, recall that one may generate a Haar-random unitary by first generating $m$ uniform random vectors in $\mathbb{C}^m$ and then performing a Gram-Schmidt orthogonalization.
In particular, this means that the first column of a Haar-random unitary may be generated by generating a uniform random vector with unit norm.
This implies that the marginal distribution over any column of a unitary drawn from the Haar measure is simply the uniform distribution over unit vectors, since we argued above that all columns are drawn from the same distribution.

Now, let us examine the runtime.
The first step (line 2 of the algorithm) requires time $t = O\left( \frac{1}{\sqrt{m}} \right)$ irrespective of $U_{jj}$ because the total time for state-transfer is $\frac{1}{\omega}\left( \frac{\pi}{2} + \cos^{-1}U_{jj} \right) \leq \frac{\pi}{\omega} = \frac{\pi}{\sqrt{m-2}}$.
Next, the second step takes time $t = \frac{1}{\omega} \cos^{-1}\left(\frac{U_{1j}}{\sqrt{1-|U_{1j}^2|}} \right) = O(\frac{1}{\omega})$.
Now, 
\begin{align}
\omega &= \sqrt{1^2 + \frac{|U_{3j}|^2/(1-|U_{jj}|^2)}{|U_{2j}|^2/(1-|U_{jj}|^2)} + \frac{|U_{4j}|^2}{|U_{2j}|^2} + \ldots} \\
&= \sqrt{\frac{\sum_{i=2, i\neq j}^m |U_{ij}|^2}{|U_{2j}|^2}} = \sqrt{\frac{1-|U_{1j}|^2 - |U_{jj}|^2}{|U_{2j}|^2}}
\end{align}
Now in cases where $|U_{jj}| \leq |U_{1j}|$ (where $|U_{1j}|$ is the maximum absolute value of the column entry among all other modes $i \neq j$), which happens with probability $1- \frac{1}{m}$, we will have $\omega^2 \geq \frac{1-2|U_{1j}|^2}{|U_{2j}|^2}$.
In the other case when $|U_{jj}| \geq |U_{1j}|$, meaning that the maximum absolute value among all entries of column $j$ is in row $j$ itself, we again have $\omega^2 \geq \frac{1-2|U_{jj}|^2}{|U_{2j}|^2}$.
Both these cases can be written together as $\omega^2 \geq \frac{1-2|U_{1j}|^2}{|U_{2j}|^2}$, where we now denote $U_{1j}$ as the entry with maximum absolute value among \emph{all} elements of column $j$.
The analysis completely hinges on the typical $\omega$ we have, which in turn depends on $|U_{1j}|$.
We will show $\Pr \left( \omega^2 \geq \frac{cm}{\log m} \right) \geq 1 - \frac{1}{\poly(m)}$, which will prove the claim for $\alpha = 0$.
\begin{align}
\Pr \left( \omega^2 \geq \frac{cm}{\log m} \right) \geq \Pr \left( 1 - 2|U_{1j}|^2 \geq c_1 \  \&\ |U_{2j}|^2 \leq \frac{c_1 \log m}{c m} \right)
\end{align}
since the two events on the right hand side suffice for the first event to hold.
Further,
\begin{align}
 \Pr \left( 1 - 2|U_{1i}|^2 \geq c_1 \  \&\ |U_{2j}|^2 \leq \frac{c_1 \log m}{c m} \right) \geq \Pr(|U_{1j}|^2 \leq \frac{c_1 \log m}{c m})
\end{align}
for large enough $m$ with some fixed $c_1 = 0.99$ (say), since $|U_{2j}|^2 \leq |U_{1j}|^2$ and $1- 1.98 \log m/m \geq 0.99$ for large enough $m$.

To this end, we refer to the literature on order statistics of uniform random unit vectors $(z_1, z_2, \ldots z_m) \in \mathbb{C}^m$ \cite{Lakshminarayan2008}.
This work gives an explicit formula for $F(x,m)$, the probability that all $|z_j|^2 \leq x$. 
We are interested in this quantity at $x = c_1 \log m/(cm)$, since this gives us the probability of the desired event ($\omega^2 \geq cm/\log m$).
We have
\begin{align}
\Pr\left(\frac{1}{k+1} \leq x \leq \frac{1}{k} \right) = \sum_{l=0}^k \begin{pmatrix} m\\ l \end{pmatrix} (-1)^l (1-lx)^{m-1}.
\end{align}
It is also argued in Ref.~\cite{Lakshminarayan2008} that the terms of the series successively underestimate or overestimate the desired probability.
Therefore we can expand the series and terminate it at the first two terms, giving us an inequality:
\begin{align}
\Pr\left(\frac{1}{k+1} \leq x \leq \frac{1}{k} \right) = 1 - m(1-x)^{m-1} + \frac{m^2}{2}(1-2x)^{m-1} - \ldots \\
\geq 1 - m(1-x)^{m-1}.
\end{align}
Choosing $c=c_1/4 = 0.2475$, we are interested in the quantity when $k = \lfloor \frac{m}{4\log m} \rfloor$:
\begin{gather}
\Pr (x \leq 4 \log m/m) \geq 1 - m (1-4\log m/m)^{m-1} \geq 1 - \frac{1}{m^{3-4/m}}, \mathrm{since \ } \\
(1-4\log m/m)^{m-1} = \exp \left[ (m-1) \log \left( 1-\frac{4\log m}{m} \right) \right] \leq \exp \left[ -4(m-1) \frac{\log m}{m} \right] = m^{-4(1-1/m)}.
\end{gather}
This implies that the time for the single-shot protocol is also $t = O(\frac{1}{\omega}) = O(\sqrt{\frac{\log m}{m}})$ for a single column.
Notice that we can make the polynomial appearing in $\Pr(\omega^2 \geq cm/\log m) \geq 1 - 1/\poly (m)$ as small as possible by suitably reducing $c$.
To extend the proof to all columns, we use the union bound.
In the following, let $t_j$ denote the time to implement column $j$.
\begin{align}
 \Pr(\exists j: t_j > \sqrt{\frac{\log m}{cm}}) &\leq \sum_j \Pr(t_j > \sqrt{\frac{\log m}{cm}}) \\
 & \leq  m \times \frac{1}{\poly(m)} = \frac{1}{\poly(m)}
\end{align}
when the degree in the polynomial is larger than 1, just as we have chosen by setting $c=0.2475$.
This implies
\begin{align}
 \Pr (\forall j: t_j \leq \sqrt{\frac{\log m}{cm}}) = 1 - \Pr(\exists j: t_j > \sqrt{\frac{\log m}{cm}}) \geq 1 -\frac{1}{\poly(m)}.
\end{align}
This completes the proof in the case $\alpha =0$.
When $\alpha \neq 0$, we can in the worst-case set each coupling constant to a maximum of $O(m^{-\alpha/D})$, which is the maximum coupling strength of the furthest two sites separated by a distance $O(m^{1/D})$.
This factor appears in the total time for both the state-transfer \cite{Guo2019a} and single-shot protocols, and simply multiplies the required time, making it $O\left(\sqrt{\frac{\log m}{m}}\times m^{\alpha/D}\right) = O\left(\frac{\sqrt{\log m}}{m^{1/2 - \alpha/D}}\right)$.
Finally, if there are any phase shifts that need to be applied, they can be achieved through an on-site term $J_{ii} a^\dag_i a_i$, whose strength is unbounded by assumption and can thus take arbitrarily short time.
\end{proof}
The total time for implementing boson sampling on $n$ bosons is therefore $O\left(n\frac{\sqrt{\log m}}{m^{1/2 - \alpha/D}}\right) = \tilde{O}\left(n^{1 + \beta (\frac{\alpha}{D}-\frac{1}{2})}\right)$, since we should implement $n$ columns in total.

\subsection{IV.A: \qquad Optimizing hardness time}
We can optimize the hardness time by implementing boson sampling not on $n$ bosons, but on $n^\delta$ of them, for any $\delta \in (0,1]$.
The explicit lower bounds on running time of classical algorithms we would get assuming fine-grained complexity-theoretic conjectures is again something like $\exp[n^{\poly(\delta)}]$ for any $\delta \in (0,1]$.
This grows very slowly with $n$, but it still qualifies as subexponential, which is not polynomial or quasipolynomial (and, by our definition, would fall in the category ``hard").
This choice of parameters allows us to achieve a smaller hardness timescale at the cost of getting a coarse (type-II) transition.
We analyze this idea in three cases: $\alpha \leq D/2$, $\alpha \in (\frac{D}{2}, D]$ and $\alpha > D$.

When $\alpha \leq D/2$, we perform boson sampling on the nearest set of $n^\delta$ bosons with the rest of the empty sites in the lattice as target sites.
In terms of the linear optical unitary, the unitary acts on $m - n^\delta = \Theta(m)$ sites in the lattice, although only the $n^\delta$ columns corresponding to initially occupied sites are relevant.
Using the protocol in \cref{lem_fastcol}, the total time to implement $n^\delta $ columns of an $m \times m$ linear optical unitary is $O(n^\delta m^{\alpha/D-1/2} \log n) = \tilde{O}(n^{\delta} n^{\frac{\beta}{D}{(\alpha - D/2)}})$.

When $\alpha \in (\frac{D}{2},D]$, the strategy is modified.
We first move the nearest set of $n^\delta$ bosons into a contiguous set of sites within a single cluster.
This takes time $O(n^\delta)$, since each boson may be transferred in time $O(1)$.
We now perform boson sampling on these $n^\delta$ bosons with the surrounding $n^{2\delta}$ sites as targets, meaning that the effective number of total sites is $m_\mathrm{eff} = O(n^{2\delta})$, as required for the hardness of boson sampling.
Applying \cref{lem_fastcol}, the time required to perform hard instances of boson sampling is now $O(n^\delta n^{2\delta (\alpha/D - 1/2)} \log n) = n^{O(\delta)}$ for arbitrarily small $\delta > 0$.

Lastly, when $\alpha > D$, we use the same protocol as above.
The time taken for the state transfer is now $n^\delta \times \min[L,L^{\alpha -D}]$.
Once state transfer has been achieved, we use nearest-neighbor hops instead of \cref{lem_fastcol} to create an instance of boson sampling in time $O(n^{2\delta/D})$.
Since state transfer is the limiting step, the total time is $n^\delta \times \min[L,L^{\alpha -D}]$.
The hardness timescale is obtained by taking the optimum strategy in each case, giving the hardness timescale $t_\mathrm{hard} = \tilde{O}(n^{\gamma^\mathrm{II}_\mathrm{hard}})$, where
\begin{align}
 \gamma^\mathrm{II}_\mathrm{hard} = \delta +
 \begin{cases}
  \frac{\beta-1}{D} \min[1,\alpha -D] & \alpha > D \\
  0 & \alpha \in (\frac{D}{2},D] \\
  \frac{\beta}{D} \left(\alpha -\frac{D}{2}\right) & \alpha < \frac{D}{2}
 \end{cases}
\end{align}
for an arbitrarily small $\delta>0$.
This proves \cref{thm_hard} for free bosons and for interacting bosons in the case $\alpha < D/2$.
When we compare with Ref.~\cite{Deshpande2018}, which states a hardness result for $\alpha \rightarrow \infty$, we see that we have almost removed a factor of $n$ from the timescale coming from implementing $n$ columns of the linear optical unitary.
Our result here gives a coarse hardness timescale of $\Theta(L)$ that matches the easiness timescale of $L$.
More importantly, this makes the noninteracting hardness timescale the same as the interacting one.

\end{widetext}


\begin{thebibliography}{70}
\bibitem [{note_littleo()}]{note_littleo}%
  \BibitemOpen
  \bibinfo {note} {$f= o(g)$ means $f/g \rightarrow 0$ as $n \rightarrow
  \infty$, while $f = \omega(g)$ is equivalent to $g = o(f)$. When $f/g
  \nrightarrow \infty$, we say $f=O(g)$, and similarly $f=O(g)
  \Longleftrightarrow g = \Omega(f)$. Finally, $f=\Theta(g)$ if $f=O(g)$ and
  $f=\Omega(g)$. The precise asymptotic dependence on $n$ can be
  arbitrary.}\BibitemShut {Stop}%
\bibitem [{note_weaker()}]{note_weaker}%
  \BibitemOpen
  \bibinfo {note} {This is a weaker requirement than demanding a classical
  sampler that works for arbitrary $\epsilon $ and has runtime $\protect
  \mathrm {poly}(n,1/\epsilon )$.}\BibitemShut {Stop}%
\bibitem [{note_precisedef()}]{note_precisedef}%
  \BibitemOpen
  \bibinfo {note} {More precisely, for sharp transitions, we have
  $t_\mathrm{hard} = \Theta(t_\mathrm{easy})$, while for coarse transitions,
  $t_\mathrm{hard} = \omega(t_\mathrm{easy})$. See \cite{note_littleo} for
  definitions of $O,\Omega,\Theta,o$, and $\omega$ notations.}\BibitemShut
  {Stop}%
\bibitem [{note_coarsecaveat()}]{note_coarsecaveat}%
  \BibitemOpen
  \bibinfo {note} {Since coarseness and sharpness are defined with respect to
  the optimal exponents and we have not proved the optimality of the derived
  exponents, we cannot definitively say from our results that the transition is
  coarse.}\BibitemShut {Stop}%
\end{thebibliography}
\end{document}